\documentclass[12pt]{article}
\usepackage{mathpazo}
\usepackage{lineno}
\usepackage{amssymb}
\usepackage{graphicx}
\usepackage{amsfonts}
\usepackage{amsmath}
\usepackage{mathrsfs}
\usepackage{amsthm}
\usepackage{mathtools}
\usepackage{ascmac}
\usepackage{bm}
\usepackage{color}
\usepackage{pgfplots}
\pgfplotsset{compat=1.16}
\usepackage{here}
\usepackage{multirow}
\usepackage[hang,small,bf]{caption}
\usepackage[subrefformat=parens]{subcaption}
\usepackage{authblk}
\usepackage[colorlinks=true, citecolor=teal, linkcolor=teal, urlcolor=teal]{hyperref}
\usepackage[authoryear]{natbib}
\usepackage{threeparttable}
\usepackage{comment}
\usepackage{stackengine}
\usepackage{tikz}
\usetikzlibrary{positioning, calc}
\usetikzlibrary{arrows.meta}
\usetikzlibrary{bending}
\usepackage{setspace}
\usepackage{pdflscape}

\def\E#1{\mathbb{E}\left[#1\right]}
\def\V#1{\mathbb{V}\left[#1\right]}
\def\P#1{\mathbb{P}\left[#1\right]}
\def\C#1{\mathbb{C}\mathrm{ov}\left[#1\right]}

\newcommand{\argmin}{\mathop{\rm arg\,min}\limits}

\newcommand{\indep}{\perp\!\!\!\!\perp} 

\numberwithin{equation}{section}

\newtheorem{assumption}{Assumption}

\newtheorem{lemmax}{Lemma}
\newenvironment{lemma}
  {\pushQED{\qed}\lemmax}
  {\popQED\endlemmax}
\newtheorem{propx}{Proposition}
\newenvironment{prop}
  {\pushQED{\qed}\propx}
  {\popQED\endpropx}
  
\theoremstyle{definition}

\makeatletter
\renewenvironment{proof}[1][\proofname]{%
  \par\pushQED{\qed}\normalfont%
  \topsep6\p@\@plus6\p@\relax
  \trivlist\item[\hskip\labelsep\bfseries#1\@addpunct{.}]%
  \ignorespaces
}{%
  \popQED\endtrivlist\@endpefalse
}
\makeatother



\makeatletter
\renewcommand*{\@fnsymbol}[1]{\ensuremath{\ifcase#1\or \flat\or * \else\@ctrerr\fi}}
\makeatother

\title{\bf Joint Inference for the Regression Discontinuity Effect and Its External Validity}
\author{Yuta Okamoto\thanks{
\href{mailto:yuta.okamoto1998@outlook.com}{yuta.okamoto1998@outlook.com}}}
\affil{Graduate School of Economics, Kyoto University}

\begin{document}
\maketitle
\begin{abstract}
    The external validity of regression discontinuity designs is crucial for informing policy but is rarely examined in applied work.
    To advance empirical practice, we propose a joint inference procedure for the treatment effect and its local external validity, captured by the treatment effect derivative (TED), within a robust bias correction framework. 
    We further introduce a locally linear treatment effects assumption, which extends the scope of the TED and enables identification and the construction of a uniform confidence band for extrapolated effects. These methods apply to most empirical studies. Empirical illustrations demonstrate their practical usefulness.
\end{abstract}

\vspace{0.2cm}


\vspace{0.2cm}


\newpage
\doublespacing
\section{Introduction}\label{sec: introduction}
In causal analysis, internal validity and external validity are two central concerns. 
Internal validity is a necessary condition for drawing causal conclusions, but external validity is equally essential for informing policy. Thus, internally valid causal analysis alone is not sufficient, as emphasized in \citet[Chapter 8]{Duflo_etal:2007} and \cite{Vivalt:2020}. 
Nevertheless, external validity is often overlooked, as \cite{Peters_etal:2016} report for randomized controlled trials, suggesting that empirical studies rarely address this issue explicitly.

Regression discontinuity (RD) designs, one of the most popular and credible empirical strategies, face a similar, if not more severe, challenge. In the RD framework, treatment status changes discontinuously at a cutoff point, enabling identification of the treatment effect at the threshold under a mild smoothness assumption \citep{Hahn_etal:2001}. This attractive feature, however, comes at the cost of uncertain external validity. Standard RD designs provide no information about treatment effects away from the cutoff, thereby limiting the generalizability of the RD estimates regardless of their statistical significance.
For example, Figure \ref{fig: ex} illustrates two scenarios in which the RD effects at the cutoff are identical, yet their policy implications differ. In Figure \ref{fig: ex1}, the RD effect plausibly extends beyond the cutoff, whereas in Figure \ref{fig: ex2} it may be questionable whether one can conclude with confidence that the treatment has a positive effect for everyone. 
Thus, focusing solely on the treatment effect at the cutoff often leaves the policy implications unclear.
\begin{figure}[t]
    \centering
    \begin{subfigure}[b]{0.45\textwidth}
        \centering
        \includegraphics[width=0.8\linewidth]{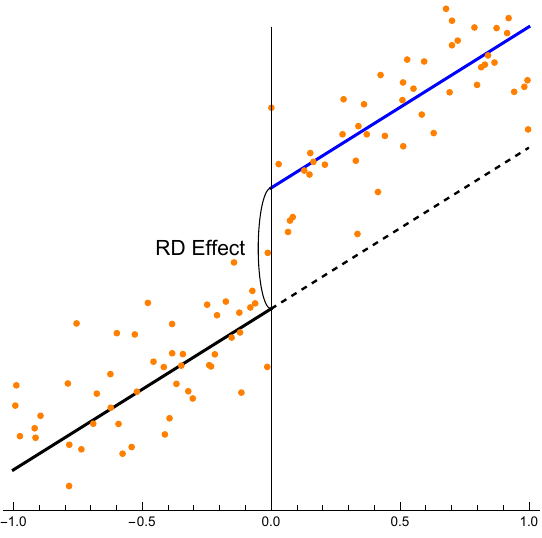}
        \caption{}
        \label{fig: ex1}
    \end{subfigure}
    \hfill
    \begin{subfigure}[b]{0.45\textwidth}
        \centering
        \includegraphics[width=0.8\linewidth]{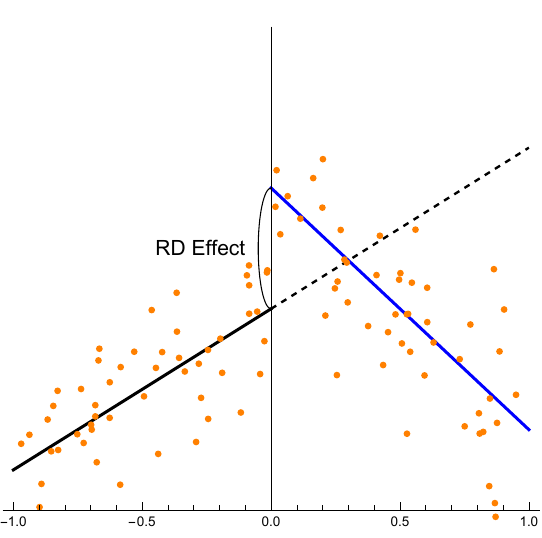}
        \caption{}
        \label{fig: ex2}
    \end{subfigure}
    \caption{Same RD Effect, Different Policy Implication}
    \label{fig: ex}
\end{figure}

Despite this crucial limitation, we find only 26 papers that discuss external validity among 70 empirical RD studies published in leading economic journals between 2020 and 2024 (details are provided in Section \ref{sec: survey}). 
Furthermore, only four studies cite recent theoretical proposals for extrapolation. 
These results may suggest not only that researchers devote limited attention to external validity, but also that extrapolation strategies proposed in theoretical econometrics are not necessarily well aligned with every empirical context, because such proposals typically require stronger assumptions, informative covariates, or design features such as multiple cutoff points \citep{Angrist_Rokkanen:2015jasa, Bertanha_Imbens:2020jbes, Cattaneo_etal:2021JASA_extrapolating, Deaner_Kwon:2025, Okamoto_Ozaki:2025}. 
This observation motivates us to develop a simple, practically accessible econometric tool.  

To this end, this article develops a simple joint inference procedure by incorporating the local extrapolation method of \cite{Dong_Lewbel:2015}---reviewed in the following paragraph---into the widely used robust bias correction (RBC) methods of \cite{Calonico_etal:2014}.
The \citeauthor{Dong_Lewbel:2015} approach is particularly attractive because it does not require additional covariates or design features and imposes only a weak differentiability assumption; in other words, it can be applied to most empirical settings.
Embedding it in the RBC framework ensures compatibility with standard practice: our confidence region is a natural extension of methods already familiar to empirical researchers, remains accessible, and requires no changes to the research design (including bandwidth choice).

To make our proposal concrete, we begin by introducing the \citeauthor{Dong_Lewbel:2015} strategy.
Let $\tau(x)$ denote the treatment effect function with $x=0$ as the cutoff. 
Whereas most RD studies focus on identifying and inferring $\tau(0)$, \cite{Dong_Lewbel:2015} proposes examining $\tau^{\prime}(0)$, since this derivative contains information about the local extrapolatability of the RD effect $\tau(0)$. 
For instance, if $\tau^{\prime}(0)=0$, the estimated treatment effect plausibly extends beyond the cutoff, as illustrated in Figure \ref{fig: ex1}. 
By contrast, if $\tau^{\prime}(0)<0$, the treatment effect may decline at other values of the running variable, raising concerns about the generalizability of the RD effect, as shown in Figure \ref{fig: ex2}.

In this respect, it is evident that the policy implications of RD studies become clearer when both $\tau(0)$ and $\tau^{\prime}(0)$ are reported.
Unfortunately, however, none of the surveyed studies report the derivative.
With the aim of encouraging broader adoption of the \citeauthor{Dong_Lewbel:2015} strategy, this article offers two complementary contributions. The first is a joint inference procedure for the pair $(\tau(0), \tau^{\prime}(0))$, obtained by deriving their joint asymptotic distribution under the RBC framework of \cite{Calonico_etal:2014}.
This RBC-based confidence region should be familiar and accessible to applied researchers and is therefore well-positioned for widespread adoption in empirical work.

Although it is an essential first step, the local external validity indicator $\tau^{\prime}(0)$ may not provide clear-cut policy implications, particularly when it is non-zero, compared with a more comprehensive extrapolation analysis such as \cite{Cattaneo_etal:2021JASA_extrapolating}. 
For instance, what does $\tau^\prime(0)=0.5$ imply about treatment effects away from the cutoff point? 
Our second contribution is therefore to introduce a \textit{locally linear treatment effects assumption}, which complements and extends the scope of \citeauthor{Dong_Lewbel:2015}'s treatment effect derivative $\tau^{\prime}(0)$ by enabling more direct extrapolation analysis of RD effects.
Under this additional assumption, RD effects away from the cutoff are characterized by the linear function, $\tau(x) = \tau(0) + \tau^{\prime}(0)x$.
This locally linear assumption and resulting identification are \textit{more general} than simply extending the standard RD effect, which effectively restricts $\tau^{\prime}(0)$ to be zero.
Moreover, the limit distribution used to derive the joint confidence region can also be employed to construct a uniform confidence band for the extrapolated RD effects, thereby enabling researchers to conduct a more informative extrapolation analysis that goes beyond merely assessing external validity.
Furthermore, we show that the confidence band remains valid for any subregion where the local linearity assumption holds, even if it does not hold over the entire initial interval.
That is, even when readers (or policymakers) disagree with researchers about the region over which the assumption is plausible, they can still perform a valid inference using the reported confidence band for any subinterval.
Thanks to this property, our approach does not raise concerns about specification search or ad-hoc tuning, thereby preserving the hallmark transparency of RD designs.

Our proposed approaches require no additional covariates or design features, making them applicable to virtually all RD settings and thereby enhancing the policy relevance of many empirical studies. 
\textcolor{black}{We illustrate their usefulness with two empirical applications. In the first, we find evidence of external validity and homogenous treatment effects even away from the cutoff. In the second, local external validity is rejected, suggesting that the standard RD effect should be interpreted with caution; moreover, our extrapolation analysis reveals treatment-effect heterogeneity over the running variable.}

\subsection*{Related Literature}
RD designs are widely acknowledged as credible quasi-experimental tools, both empirically and theoretically.
Empirically, \cite{Brodeur_etal:2020} finds that RD papers are less prone to $p$-hacking and publication bias. \cite{Hyytinen_etal:2018} provides empirical evidence showing that the RBC procedure of \cite{Calonico_etal:2014} yields inference results that are consistent with experimental estimates at the cutoff. Both works suggest strong internal validity of RD studies.
Theoretically, a growing body of research has proposed methods to enhance the internal validity of RD studies.
A very short list of recent contributions includes the sharp validity test for fuzzy RD designs by \cite{Arai_etal:2022}, the manipulation-robust bounding approach of \cite{Gerard_etal:2020}, the unified falsification test of \cite{Fusejima_etal:2025}, the analysis of measurement error by \cite{Dong_Kolesar:2023}, and adjustments for discrete running variables proposed by \cite{Dong:2015} and \cite{Kolear_Rothe:2018}.
See \cite{Cattaneo_Titiunik:2022} for a more comprehensive literature review.

In contrast, how carefully external validity is assessed in empirical RD studies remains not fully understood.
Theoretically as well, the literature is still developing, as highlighted by \cite{Abadie_Cattaneo:2018}: ``While there is some very recent work that addresses the issue of extrapolation of RD treatment effects [...], more work is certainly needed, covering both methodological and empirical issues."
Recent contributions, other than \cite{Dong_Lewbel:2015}, include conditional-independence-type extrapolation methods proposed by \cite{Angrist_Rokkanen:2015jasa} and \cite{Bertanha_Imbens:2020jbes}, \textcolor{black}{as well as approaches that exploit additional covariates or multiple cutoff points, such as \cite{Cattaneo_etal:2016jop}, \cite{Cattaneo_etal:2021JASA_extrapolating}, \cite{Deaner_Kwon:2025}, \cite{Okamoto_Ozaki:2025}, and \cite{Sun:2023}. 
Relatedly, \cite{Bertanha:2020} develops a method for estimating average treatment effects under counterfactual treatment-assignment and dose policies in settings with many cutoff points.}
Furthermore, some studies investigate treatment effect heterogeneity with respect to predetermined covariates (\citealp{Hsu_Shen:2019, Hsu_Shen:2021}).
\textcolor{black}{
The present article contributes to this second line of literature by providing suggestive survey evidence on recent empirical practice and developing theoretical results for external validity analysis and extrapolation exercises, integrating the seminal yet distinct contributions of \citet{Calonico_etal:2014} and \citet{Dong_Lewbel:2015}. 
Our locally linear treatment effects assumption has been used in \citet{Dong_Lewbel:2015} and \citet{Cerulli_etal:2017} without being explicitly stated and is therefore not entirely novel; our contribution is to enhance its practical applicability by providing a uniform inference procedure.
}

\subsection*{Plan of the Article}
Section \ref{sec: survey} reviews recent empirical studies to examine the current state of internal and external validity analysis in RD designs.
Section \ref{sec: sharp} presents the main results of the paper, focusing on the sharp RD design. After introducing the setup and notation, Section \ref{subsec: inference} develops a robust local linear confidence region, and Section \ref{subsec: linear} proposes an extrapolation method.
Section \ref{sec: empirical} reports empirical applications to highlight the potential usefulness of our proposed procedure.
We also provide simulation studies in Section \ref{sec: simulation} to investigate the finite sample properties of our proposed procedures.
Finally, Section \ref{sec:conclusion} concludes the paper with practical recommendations. All proofs, along with the generalization to local polynomial confidence regions and the extension to the fuzzy RD design, are provided in the Online Appendix.

\section{Empirical Practices in Recent RD Studies}\label{sec: survey}
\subsection{Sample and Procedure}
We begin by investigating how recent RD studies address issues of internal and external validity.
We review empirical studies published between 2020 and 2024 in the following leading economics journals: \textit{American Economic Journal: Applied Economics} (AEJ: Applied), \textit{American Economic Journal: Economic Policy} (AEJ: EP), \textit{American Economic Review} (AER), \textit{American Economic Review: Insights} (AERI), \textit{Journal of Human Resources} (JHR), and \textit{Journal of Labor Economics} (JOLE).
We identified articles containing the phrase ``regression discontinuity" using the EBSCOhost database. After excluding theoretical contributions, we obtain a sample of 70 empirical studies.

To systematically assess whether these studies discuss the external and internal validity of their RD design, and to avoid subjective judgment, we apply the following criteria.  
External validity is identified by (i) the presence of keywords such as 
``external validity," ``extrapolat*," or ``transportability" (where the asterisk (*) denotes a wildcard that allows for variations of the root word, e.g., \textit{extrapolate}, \textit{extrapolation}), and (ii) citations to leading contributions including \cite{Angrist_Rokkanen:2015jasa}, 
\cite{Bertanha_Imbens:2020jbes}, 
\cite{Cattaneo_etal:2021JASA_extrapolating}, and 
\cite{Dong_Lewbel:2015}.  
Internal validity is assessed using a parallel method. However, discussions of internal validity seldom appear under the explicit phrase ``internal validity." Therefore, in addition to this phrase, we also search for related keywords---such as ``continuity," ``manipulat*," and ``placebo"---that suggest a discussion of identification assumptions. We further record citations to \cite{Calonico_etal:2014} and \cite{McCrary:2008}, which provide valid inference procedures and a diagnostic test that are essential for internally valid RD analysis.  
We apply these criteria by using text analysis in R \citep{R}, which allows us to reproducibly detect whether each article includes the specified keywords or citations.\footnote{\textcolor{black}{We identify discussions of internal and external validity using fixed keywords and citations, which may under-detect conceptual discussions phrased differently or without the designated citations. Our counts should thus perhaps be interpreted as a conservative lower bound on explicit discussion.}}

\subsection{Results}
The survey results are summarized in Table \ref{tab: empirical practice}.
We find that almost all RD studies carefully address the internal validity of their analysis.
By contrast, far fewer articles explicitly discuss external validity.
In total, only 26 out of 70 articles contain the keywords related to external validity, and only four cite the leading contributions in this area.
\begin{table}
    \begin{center}
    \begin{tabular}{lcccccc}
    \hline\hline
        \multirow{2}{*}{Journal} & \multirow{2}{*}{\# of Articles} & \multicolumn{2}{c}{Internal Validity}  &  \multicolumn{2}{c}{External Validity} & \multirow{2}{*}{TED}\\\cline{3-4} \cline{5-6}
         & &  Keywords & Citation & Keywords & Citation & \\\hline
        AEJ: Applied & 16 & 16 & 13 & 7 & 1 & 0\\
        AEJ: EP & 16 & 15 & 15 & 4 & 1 & 0\\
        AER/AERI & 10 & 8 & 8 & 4 & 0 & 0\\
        JHR & 22 & 22 & 15 & 9 & 1 & 0\\
        JOLE & 6 & 6 & 6 & 2 & 1 & 0\\\hline
        Total & 70 & 67 & 57 & 26 & 4 & 0\\\hline
    \end{tabular}
    \caption{Reporting of Internal and External Validity in Empirical RD Studies}
    \label{tab: empirical practice}
    \end{center}

    \footnotesize
    \renewcommand{\baselineskip}{11pt}
    \textbf{Note:} The table reports, for each journal, the number of empirical RD studies published between 2020 and 2024 (column 2). 
    Columns 3-4 show the number of articles that mention at least one of the internal validity keywords or cite at least one of the designated references (\citealp{Calonico_etal:2014, McCrary:2008}). 
    Columns 5-6 report the corresponding counts for external validity, based on the presence of keywords or citations to leading contributions 
    (\citealp{Angrist_Rokkanen:2015jasa, Bertanha_Imbens:2020jbes, 
    Cattaneo_etal:2021JASA_extrapolating, Dong_Lewbel:2015}). 
    The last column reports the number of studies that explicitly present 
    the treatment effect derivative (TED) following \cite{Dong_Lewbel:2015}.
\end{table}

This pattern may suggest two possible implications.
First, the relatively small number of citations compared to keyword mentions indicates that most existing proposals are not necessarily directly applicable to many empirical RD settings, so that discussions of external validity are often conducted on a case-by-case basis.
Second, compared to internal validity, researchers devote relatively little attention to external validity. This is evident from the fact that the procedure of \cite{Dong_Lewbel:2015} has not been widely adopted in empirical practice, despite its broad applicability.

This issue is of particular concern. Given the inherently local nature of RD designs, the generalizability of the estimated effect is often unclear.
To improve the clarity of their policy implications, it would be desirable to have a simple tool that naturally extends methods routinely used in empirical research. In the next section, we propose such a procedure.

\section{External Validity and Extrapolation under Sharp RD}\label{sec: sharp}
To advance empirical practice, this section proposes a simple inference procedure that builds on \cite{Dong_Lewbel:2015}.
To encourage broader adoption of their strategy, we embed it in the widely used RBC confidence intervals of \cite{Calonico_etal:2014}.
This integration ensures close compatibility with standard practice, making the procedure both accessible to applied researchers and implementable without changes to the research design, including bandwidth choice.

Extending this idea, we also introduce an additional layer of assumptions that enables researchers to conduct a more comprehensive extrapolation analysis, rather than only assessing external validity. 
The proposed strategy again does not require additional covariates or design features, and is therefore applicable to a wide range of RD studies. 

\subsection{Setup}
We focus on the sharp RD design, while the extension to the fuzzy RD case is provided in the Online Appendix. 
We observe a random sample ${(Y_i(0), Y_i(1), X_i): i=1,\ldots,n}$.
Here, $Y_i(1)$ denotes the potential outcome under treatment, and $Y_i(0)$ the potential outcome under control.
$X_i$ is the running variable with density $f(x)$.
Unit $i$ receives the treatment if $X_i \geq 0$ and is untreated otherwise.
Hence, the observed outcome is $Y_i = Y_i(0)\mathbf{1}\left\{X_i < 0\right\} + Y_i(1)\mathbf{1}\left\{X_i \geq 0\right\}$.

Define the treatment effect function as $\tau(x) \coloneqq \E{Y_i(1) - Y_i(0) | X_i=x}$.
Then, the sharp RD effect is given by $\tau_{\mathtt{SRD}} \coloneqq \tau(0)$.
To assess a local external validity, \cite{Dong_Lewbel:2015} introduces the sharp RD treatment effect derivative (TED), which is given by $\tau_{\mathtt{SRD}}^\prime \coloneqq \tau^\prime(0)$. 
Under a weak differentiability assumption, these are identified as 
\begin{align*}
    \tau_{\mathtt{SRD}} &= \mu_+ - \mu_-,\quad
    \mu_+ \coloneqq \lim_{x\downarrow0}\mu(x),\quad
    \mu_- \coloneqq \lim_{x\uparrow0}\mu(x),\\
    \tau_{\mathtt{SRD}}^\prime &= \mu_+^\prime - \mu_-^\prime,\quad
    \mu_+^\prime \coloneqq \lim_{x\downarrow0}\mu^\prime(x),\quad
    \mu_-^\prime \coloneqq \lim_{x\uparrow0}\mu^\prime(x),
\end{align*}
where $\mu(x)\coloneqq \E{Y_i | X_i=x}$ \citep{Dong_Lewbel:2015}.
They also show that both $\tau_{\mathtt{SRD}}$ and $\tau_{\mathtt{SRD}}^\prime$ can be estimated from standard local linear regression, but their analysis does not extend to a joint inference procedure for the pair $(\tau_{\mathtt{SRD}}, \tau_{\mathtt{SRD}}^\prime)$.
Our first goal is to perform valid inference on $\tau_{\mathtt{SRD}}$ and $\tau_{\mathtt{SRD}}^\prime$ jointly.

\subsection{Robust Local Polynomial Confidence Region}\label{subsec: inference}
We begin by introducing maintained assumptions, which are fairly general and standard in the literature (\citealp{Calonico_etal:2014}):
\begin{assumption}\label{assumption: regularity}
    In a neighbourhood $(-\kappa_0,\kappa_0)$ around the cutoff:
    \item[(i)] $\E{Y_i^4 | X_i =x}$ is bounded, and $f(x)$ is continuous and bounded away from zero.
    \item[(ii)] $\mu_+(x)\coloneqq\E{Y_i(1)|X_i=x}$ and $\mu_-(x)\coloneqq\E{Y_i(0)|X_i=x}$ are $S$-times continuously differentiable.
    \item[(iii)] $\sigma_+^2(x)\coloneqq \V{Y_i(1)|X_i=x}$ and $\sigma_-^2(x)\coloneqq \V{Y_i(0)|X_i=x}$ are continuous and bounded away from zero.
\end{assumption}
\begin{assumption}\label{assumption: kernel}
    The kernel function $K$ is a symmetric, nonnegative, bounded, and continuous second-order kernel with compact support on $[-1,1]$.
\end{assumption}

Following the RD literature (\citealp{Calonico_etal:2014, Hahn_etal:2001}), we employ the local linear regression, while a more general treatment is given in the Online Appendix, Section S2.1.
Specifically, we estimate $\hat{\tau}_{\mathtt{SRD}} = \hat{\mu}_+ - \hat{\mu}_-$ and $\hat{\tau}_{\mathtt{SRD}}^\prime = \hat{\mu}_+^\prime - \hat{\mu}_-^\prime$, where
\begin{align*}
    \left(\hat{\mu}_+ , \hat{\mu}_+^\prime\right)^\top
    = \argmin_{(b_0,b_1)^\top \in\mathbb{R}^2} \sum_{i=1}^{n}\mathbf{1}\left\{X_i \geq 0\right\} \left(Y_i - b_0 - b_1 X_i\right)^2 \frac{1}{h}K\left(\frac{X_i}{h}\right),\\
    \left(\hat{\mu}_- , \hat{\mu}_-^\prime\right)^\top
    = \argmin_{(b_0,b_1)^\top \in\mathbb{R}^2} \sum_{i=1}^{n}\mathbf{1}\left\{X_i < 0\right\} \left(Y_i - b_0 - b_1 X_i\right)^2 \frac{1}{h}K\left(\frac{X_i}{h}\right).
\end{align*}
\textcolor{black}{
We here note that the \textit{simultaneous} estimation of the mean and first-derivative functions is attractive because the leading biases of $\hat{\mu}_+$ and $\hat{\mu}_+^\prime$ depend on the same quantity $\mu_+^{(2)}$ (e.g., \citealp[Theorem 3.3]{Fan_Gijbels:1996}). Hence, only a single quantity needs to be estimated for bias correction.\footnote{\textcolor{black}{Since we also care about the first derivative, a local quadratic fit ($p=2$) may be appealing under additional smoothness assumptions. A single local quadratic regression then yields both the mean and first-derivative estimators, whose leading biases depend on the third derivative (estimable, e.g., via a local cubic fit). While appealing for asymptotic bias, higher-order polynomials can inflate variance at boundary points, so finite-sample performance need not improve (see, e.g., \citealp{Gelman_Imbens:2019}, \citealp{Pei_etal:2022}, \citealp[Remark 4]{Ruppert_Wand:1994}). Moreover, bias correction becomes more challenging because it requires estimating higher-order derivatives. We investigate finite-sample properties in Section \ref{sec: simulation}.}}
}

The conditional biases of these local linear estimators are given by
\begin{align*}
    \E{\hat{\tau}_{\mathtt{SRD}}|X_1,\ldots,X_n} - {\tau}_{\mathtt{SRD}} &=
    h^2\mathrm{B}_{\mathtt{SRD}, 0} \left(\mu_+^{(2)} - \mu_-^{(2)}\right)\left\{1+o_p(1)\right\},\\
    \E{\hat{\tau}_{\mathtt{SRD}}^\prime|X_1,\ldots,X_n} - {\tau}_{\mathtt{SRD}}^\prime &=
    h \mathrm{B}_{\mathtt{SRD},1} \left(\mu_+^{(2)} - \mu_-^{(2)}\right)\left\{1+o_p(1)\right\},
\end{align*}
where $\mathrm{B}_{\mathtt{SRD}, 0}$ and $\mathrm{B}_{\mathtt{SRD},1}$ are constants depending only on observables.
We can estimate the unknown components, $\mu_+^{(2)}$ and $\mu_-^{(2)}$, using the local quadratic estimation with the same kernel $K$ and a possibly different bandwidth $b$.
With these local quadratic estimators, we can estimate the leading bias terms by $h^2\mathrm{B}_{\mathtt{SRD}, 0} (\hat{\mu}_+^{(2)} - \hat{\mu}_-^{(2)})$ and $h\mathrm{B}_{\mathtt{SRD},1} (\hat{\mu}_+^{(2)} - \hat{\mu}_-^{(2)})$, respectively.
Then, our bias-corrected estimators for $({\tau}_{\mathtt{SRD}},{\tau}_{\mathtt{SRD}}^\prime)$ are defined as
\begin{align*}
    \tilde{\tau}_{\mathtt{SRD}} \coloneqq \hat{\tau}_{\mathtt{SRD}} - h^2\mathrm{B}_{\mathtt{SRD}, 0} \left(\hat{\mu}_+^{(2)} - \hat{\mu}_-^{(2)}\right),\,\,\text{ and }\,\,
    \tilde{\tau}_{\mathtt{SRD}}^\prime \coloneqq \hat{\tau}_{\mathtt{SRD}}^\prime - h\mathrm{B}_{\mathtt{SRD},1} \left(\hat{\mu}_+^{(2)} - \hat{\mu}_-^{(2)}\right).
\end{align*}
The following lemma derives the joint limit of these two estimators:
\begin{lemma}\label{lemma: distribution sharp}
    Suppose Assumptions \ref{assumption: regularity}-\ref{assumption: kernel} hold. If $S\geq3$, $\max\{h,b\}<\kappa_0$, $n \min\{h^5,b^5\}\times\max\{h^2,b^2\}\to0$, $n\min\{h,b\}\to\infty$, then 
    \begin{align*}
        \Omega^{-1/2}\mathrm{diag}(1,h)\tilde{\Delta}(\tau_{\mathtt{SRD}}, \tau_{\mathtt{SRD}}^\prime)\to_d
        \mathcal{N}\left((0,0)^\top, \mathrm{diag}(1,1)
        \right),
    \end{align*}
    provided that $\Omega$ is invertible, where
    \begin{align*}
        \tilde{\Delta}(t, t^\prime)\coloneqq\begin{pmatrix}
                \tilde{\tau}_{\mathtt{SRD}} - t\\
                \tilde{\tau}_{\mathtt{SRD}}^\prime - t^\prime
            \end{pmatrix},\,\,
        \Omega\coloneqq\begin{pmatrix}
            \mathrm{V}_{\mathtt{SRD}} & \mathrm{C}_{\mathtt{SRD}}\\
            \mathrm{C}_{\mathtt{SRD}} & \mathrm{V}_{\mathtt{SRD}}^\prime
        \end{pmatrix},
    \end{align*}
    and $\mathrm{V}_{\mathtt{SRD}}$, $\mathrm{V}_{\mathtt{SRD}}^\prime$, and $\mathrm{C}_{\mathtt{SRD}}$ are provided in the Online Appendix Section S2.1.
    Furthermore, $\Omega$ is asymptotically invertible.
\end{lemma}
This is just a straightforward extension of the main theorem of \cite{Calonico_etal:2014}.
Importantly, there is \textit{no} additional unknown components in $\Omega$ compared to the asymptotic variance of \cite{Calonico_etal:2014}.
Hence, the same variance estimation procedures proposed in \citet[Section 5]{Calonico_etal:2014} can be employed in our case. 
Furthermore, the rate conditions on the bandwidths are the same as those made in \cite{Calonico_etal:2014}. 
Hence, we can still use the usual MSE-optimal bandwidth for the treatment effect $\tau_{\mathtt{SRD}}$ as the main bandwidth $h$, and that for $\mu_+^{(2)} - \mu_-^{(2)}$ as the bandwidth $b$ for bias estimation, following \citet[Section 4]{Calonico_etal:2014}. That is, an additional bandwidth selection is not needed \textcolor{black}{for valid inference.}\footnote{\textcolor{black}{Although the MSE-optimal bandwidth for $\tau_{\mathtt{SRD}}$ is the option most consistent with standard RD practice, other choices are possible. For example, the asymptotically optimal bandwidth that minimizes $w\mathrm{MSE}[\hat{\tau}_{\mathtt{SRD}}] + (1-w)\mathrm{MSE}[\hat{\tau}_{\mathtt{SRD}}^\prime]$ for $w\in[0,1)$ coincides with the MSE-optimal bandwidth for the derivative $\tau_{\mathtt{SRD}}^\prime$, which is also valid. See also Section \ref{sec: remarks} for an additional remark and Section \ref{sec: simulation} for a simulation comparison.}}

We are now in a position to present the joint inference procedure.
Let $\hat{\Omega}$ be a consistent estimator, and define $\hat{\Omega}_h\coloneqq \mathrm{diag}(1,h^{-1})\hat{\Omega}\mathrm{diag}(1,h^{-1})$.
Due to Lemma \ref{lemma: distribution sharp}, we obtain that $\tilde{{\Delta}}(\tau_{\mathtt{SRD}}, \tau_{\mathtt{SRD}}^\prime)^\top \hat{{\Omega}}_h^{-1} \tilde{{\Delta}}(\tau_{\mathtt{SRD}}, \tau_{\mathtt{SRD}}^\prime) \to_d \chi^2_2$.
Hence, we can construct an asymptotically valid $100(1-\alpha)\%$ confidence region for $({\tau}_{\mathtt{SRD}}, {\tau}_{\mathtt{SRD}}^\prime)$ as follows. Let $\alpha\in(0,1)$ hereafter.
\begin{prop}\label{prop: region}
    Let $\hat{{\Omega}}$ be invertible and $\hat{\Omega}\to_p\Omega$.
    Under the same assumptions in Lemma \ref{lemma: distribution sharp}, an asymptotic $1-\alpha$ confidence region of $(\tau_{\mathtt{SRD}}, \tau_{\mathtt{SRD}}^\prime)$ is given by
    \begin{align*}
        \mathcal{R}_{1-\alpha} \coloneqq \left\{(t, t^\prime)^\top\in\mathbb{R}^2:
        \tilde{{\Delta}}(t, t^\prime)^\top \hat{{\Omega}}_h^{-1} \tilde{{\Delta}}(t, t^\prime) \leq c_{1-\alpha}
        \right\},
    \end{align*}
    where $c_{1-\alpha}$ is the $100(1-\alpha)$-percentile of $\chi_2^2$ distribution.
\end{prop}
This confidence region can be represented graphically as an ellipse; see Sections \ref{sec: empirical} and \ref{sec: simulation} for illustration.
Although somewhat repetitive, it bears emphasis that no additional variance estimation or bandwidth selection is required. Consequently, researchers can incorporate this external validity assessment seamlessly into standard RD analysis without incurring any additional cost.

\subsection{Inference under a Locally Linear Treatment Effects Assumption}\label{subsec: linear}
The confidence region $\mathcal{R}_{1-\alpha}$ is itself a useful object for assessing the external validity of standard RD estimates.
However, in many applied policy settings, a more comprehensive extrapolation analysis, such as that of \cite{Cattaneo_etal:2021JASA_extrapolating}, may be more directly relevant than merely verifying external validity. 
To this end, we introduce an additional layer of assumptions, which we term the \textit{locally linear treatment effects assumption}.\footnote{The idea of imposing additional layer of assumptions is akin to the notion of layered analysis in the partial identification literature (e.g., \citealp{Manski_Nagin:1998, Manski:2011}), where progressively stronger yet still acceptable assumptions are imposed to sharpen conclusions. By analogy, in the RD setting, it seems natural to move from the identification at the cutoff under continuity, to local extrapolation under differentiability, and then to further extrapolation under local linearity.}
\begin{assumption}[Locally Linear Treatment Effects]\label{assumption: linear effects}
    For some $\delta_1, \delta_2>0$, the treatment effect function $\tau(x)$ is a linear function of $x$ over the region $[-\delta_1, \delta_2]$.
\end{assumption}
Before discussing this assumption, we provide the following identification result of the extrapolated treatment effects:
\begin{lemma}\label{lemma: identification}
    Suppose that $\mu_+(x)$ and $\mu_-(x)$ are continuously differentiable. Then, under Assumption \ref{assumption: linear effects}, $\tau(x)$ is identified over $[-\delta_1, \delta_2]$ as $\tau(x) = \tau_{\mathtt{SRD}} + \tau_{\mathtt{SRD}}^\prime x$.
\end{lemma}
Assumption \ref{assumption: linear effects} could be restrictive and may not hold in all empirical settings.\footnote{\cite{Ghosh_etal:2025} recently proposed a similar (and stronger) linear treatment effects assumption in the sharp RD setting (see their Assumption 2). However, they introduced it to improve inference on the standard RD estimand $\tau_{\mathtt{SRD}}$ and do not consider external validity or extrapolation issues. } 
Even so, it provides a natural and tractable starting point for extrapolation analysis, reflecting the common belief that treatment effects are likely to evolve smoothly. Moreover, this extrapolation result is \textit{more general} than simply extending the standard RD effect $\tau_{\mathtt{SRD}}$. 
Hence, it provides a more encompassing perspective than a sole focus on the treatment effect at the cutoff point.

Further, with the help of the additional assumption, Lemma \ref{lemma: identification} provides a way to intuitively and graphically interpret \citeauthor{Dong_Lewbel:2015}'s TED. 
Of course, the TED is a more general quantity and should not be restricted to our interpretation in every scenario. Nevertheless, our framework offers a natural lens through which their TED can be visually understood in practice, and facilitates a more informative extrapolation analysis.

Under the locally linear treatment effects assumption, we can further obtain a useful uncertainty quantification method:
\begin{prop}\label{prop: band}
    In addition to the assumptions made in Proposition \ref{prop: region}, suppose that Assumption \ref{assumption: linear effects} also holds.
    Then, $\lim_{n\to\infty}\P{\tau(x) \in \mathcal{U}_{1-\alpha}(x), \forall x\in[-\delta_1,\delta_2]} = 1-\alpha$, where
    \begin{align*}
        \mathcal{U}_{1-\alpha}(x)\coloneqq\bigg[
        &\tilde{\tau}_{\mathtt{SRD}} + \tilde{\tau}_{\mathtt{SRD}}^\prime x - c_{1-\alpha}^\star\sqrt{(1, x) \hat{\Omega}_h(1,x)^\top} ,\\
        &\qquad\tilde{\tau}_{\mathtt{SRD}} + \tilde{\tau}_{\mathtt{SRD}}^\prime x + c_{1-\alpha}^\star\sqrt{(1, x) \hat{\Omega}_h (1,x)^\top}
        \bigg],
    \end{align*}
    and $c_{1-\alpha}^\star$ is defined by $P(c_{1-\alpha}^\star)=1-\alpha$, where
    \begin{align*}
        P(s)&\coloneqq\frac{\ell}{\pi} \left[1-\exp\left(-\frac{s^2}{2}\right)\right] + 
        \frac{2}{\pi}\int_0^{\pi/2 - \ell/2}1-\exp\left(-\frac{s^2}{2\cos(u)^2}\right)\,du.
    \end{align*}
    Here, $\ell \in [0,\pi]$ denotes the smallest angle between the directions $\hat{v}(-\delta_1)$ and $\hat{v}(\delta_2)$, where $\hat{v}(x) \coloneqq \hat{\Omega}_h^{1/2}(1,x)^\top/||\hat{\Omega}_h^{1/2}(1,x)^\top||$.
\end{prop}

Although it may appear obscure at first glance, the confidence band $\mathcal{U}_{1-\alpha}(x)$ is related to the joint confidence region $\mathcal{R}_{1-\alpha}$.
Noting that $1-\exp(-s^2/2)$ is the distribution function of the standard Rayleigh distribution, it follows that $c_{1-\alpha}^\star \to \sqrt{c_{1-\alpha}}$ as $\min\{\delta_1, \delta_2\}\to\infty$, i.e., $\ell\to\pi$.
In this sense, the confidence band $\mathcal{U}_{1-\alpha}(x)$ admits an alternative interpretation: as the extrapolation interval expands, $\mathcal{U}_{1-\alpha}(x)$ converges to the \textit{projection envelope} $\mathcal{E}_{1-\alpha}(x)$ of the ellipse $\mathcal{R}_{1-\alpha}$ onto the linear functions, where
\begin{align*}
    \mathcal{E}_{1-\alpha}(x)\coloneqq\bigg[
    &\tilde{\tau}_{\mathtt{SRD}} + \tilde{\tau}_{\mathtt{SRD}}^\prime x - \sqrt{c_{1-\alpha}(1, x) \hat{\Omega}_h (1,x)^\top} ,\\
    &\qquad\tilde{\tau}_{\mathtt{SRD}} + \tilde{\tau}_{\mathtt{SRD}}^\prime x + \sqrt{c_{1-\alpha}(1, x) \hat{\Omega}_h (1,x)^\top}
    \bigg],
\end{align*}
which can be regarded as the envelope of the linear extrapolations corresponding to all $(\tilde{\tau}_{\mathtt{SRD}}, \tilde{\tau}_{\mathtt{SRD}}^\prime)$ pairs contained in $\mathcal{R}_{1-\alpha}$.\footnote{\textcolor{black}{It follows that $c_{1-\alpha}^\star \leq \sqrt{c_{1-\alpha}}$, and hence $\mathcal{U}_{1-\alpha}(x)\subseteq\mathcal{E}_{1-\alpha}(x)$. Therefore, $\mathcal{E}_{1-\alpha}(x)$ is a valid confidence band, albeit generally conservative over a finite interval.
Intuitively, the inclusion (rather than equality) reflects that different $x$ may have different least-favorable directions: over a wide extrapolation region, many pairs $(\tau,\tau^\prime)\in\mathcal{R}_{1-\alpha}$ can be least favorable for some $x$, whereas on a restricted interval only a subset of directions---formally captured by $\hat v(x)$---matters, yielding a tighter band than the one obtained by mere projection (see the proofs of Proposition~\ref{prop: band} and Lemma~\ref{lemma: crit value increasing}). 
In practice, however, the difference between the two bands may be small when the extrapolation region is only moderately wide; in our applications they are nearly indistinguishable (not reported).
}}

Moreover, we can show that $c_{1-\alpha}^\star \to \Psi^{-1}(1-\alpha/2)$ as $\max\{\delta_1, \delta_2\}\to 0$ (and hence $x\to0$), where $\Psi^{-1}$ denotes the quantile function of the standard normal distribution. Thus, $\mathcal{U}_{1-\alpha}(x)$ collapses to the usual RBC confidence interval for $\tau_{\mathtt{SRD}}$ of \cite{Calonico_etal:2014}.
That is, $\mathcal{U}_{1-\alpha}(x)$ can be viewed as a natural extension of the RBC confidence intervals to linear extrapolation analysis.

We conclude this subsection with a remark on the choice of the extrapolation window $[-\delta_1, \delta_2]$.
This choice should be guided by the researcher's expertise and research question. At first glance, this may appear to allow room for specification search, but that is not the case.
\textcolor{black}{
First, because the intercept and slope are pinned down by the estimates at the cutoff, the point estimates of the extrapolated line are invariant to the choice of the extrapolation region. Moreover, as shown below, it can be justified to select the largest interval that is plausible given the researcher's institutional knowledge and to report the confidence band over this largest region.}

Suppose an economist assumes local linearity over the interval $\mathcal{I}_{\mathtt{EC}}$ and reports the confidence band $\mathcal{U}_{1-\alpha}^{\mathtt{EC}}(x)$ on this interval. 
A policymaker (or reader), however, may believe that $\mathcal{I}_{\mathtt{EC}}$ is too wide, and that linearity holds only on a smaller interval $\mathcal{I}_{\mathtt{PM}}\subset \mathcal{I}_{\mathtt{EC}}$. 
If the policymaker is correct, then $\mathcal{U}_{1-\alpha}^{\mathtt{EC}}(x)$ is not valid over the full $\mathcal{I}_{\mathtt{EC}}$; nevertheless, it still satisfies
\begin{align*}
    \lim_{n\to\infty}\P{\tau(x) \in \mathcal{U}_{1-\alpha}^{\texttt{EC}}(x), \forall x\in\mathcal{I}_{\mathtt{PM}}} \geq 1-\alpha,
\end{align*}
which is a direct consequence of the following lemma:
\begin{lemma}\label{lemma: crit value increasing}
Let $c_{1-\alpha}^\star(\delta_1,\delta_2)$ denote the critical value defined in Proposition \ref{prop: band} for the extrapolation window $[-\delta_1,\delta_2]$. 
If $\delta_1^\prime\leq \delta_1$ and $\delta_2^\prime \leq \delta_2$, then $c_{1-\alpha}^\star(\delta_1^\prime,\delta_2^\prime) \leq c_{1-\alpha}^\star(\delta_1,\delta_2)$ holds.
That is, the critical value is nondecreasing under nested extensions of the extrapolation window.
\end{lemma}

Therefore, the confidence band $\mathcal{U}_{1-\alpha}^{\mathtt{EC}}(x)$ reported in a research article remains valid on any subinterval where Assumption \ref{assumption: linear effects} holds, even if the assumption fails on the full interval $\mathcal{I}_{\mathtt{EC}}$. 
In other words, the policymaker can rely on the reported confidence band to conduct a valid inference, \textit{without access to the underlying research data}.
This robustness justifies choosing $\delta_1$ and $\delta_2$ as large as is plausible given the researcher’s institutional knowledge, provided that the entire confidence band is transparently reported.

\subsection{Remarks}\label{sec: remarks}
\paragraph{Extrapolation Region.}
\textcolor{black}{In practice, a researcher must choose $(\delta_1,\delta_2)$, which should be guided by the research question and empirical context. That said, depending on the setting, it may be difficult to justify a particular extrapolation window, or determine the ``baseline specification" to interpret the result. In such cases, we propose considering $\delta_1=\delta_2\approx h$ as a rule of thumb. In finite samples, local linear estimation approximates $\mu_+(x)$ and $\mu_-(x)$ by linear functions over $[0,h]$ and $[-h,0]$, respectively. It seems therefore natural to adopt a linear treatment effect approximation over the same window, $[-h,h]$, as a benchmark.}

\paragraph{Testing the Local Linearity.}
\textcolor{black}{Assumption \ref{assumption: linear effects} has a testable implication $\tau^{(2)}(0)=0$. This can be tested by the local polynomial regression (see, e.g., \citealp{calonico2020optimal}).}
In cases where additional information such as multiple cutoffs is available, this assumption can be tested under alternative identification assumptions. We do such exercises in our empirical applications.

\paragraph{MISE-Optimal Bandwidth.}
\textcolor{black}{Although the MSE-optimal bandwidth for $\tau_{\mathtt{SRD}}$ is the choice most consistent with standard RD practice, one may instead be interested in the MISE-optimal bandwidth in the context of extrapolation. Note that
\begin{align*}
    &\mathrm{MISE} = \int_{-\delta_1}^{\delta_2}\E{\left\{\left(\hat{\tau}_{\mathtt{SRD}} + \hat{\tau}_{\mathtt{SRD}}^\prime x\right)-\left({\tau}_{\mathtt{SRD}} + {\tau}_{\mathtt{SRD}}^\prime x\right)\right\}^2}\,dx\\
    &= \mathrm{MSE}\left[\hat{\tau}_{\mathtt{SRD}}\right]\int_{-\delta_1}^{\delta_2}\,dx+
    \mathrm{MSE}\left[\hat{\tau}_{\mathtt{SRD}}^\prime\right]\int_{-\delta_1}^{\delta_2} x^2\,dx+
    2\E{\left(\hat{\tau}_{\mathtt{SRD}}-{\tau}_{\mathtt{SRD}}\right)\left(\hat{\tau}_{\mathtt{SRD}}^\prime-{\tau}_{\mathtt{SRD}}^\prime\right)}\int_{-\delta_1}^{\delta_2} x\,dx.
\end{align*}
By a standard Taylor-expansion argument, the second term is the dominant term as long as $\max\{\delta_1,\delta_2\}<\infty$. Hence, the asymptotically MISE-optimal bandwidth coincides with the MSE-optimal bandwidth for ${\tau}_{\mathtt{SRD}}^\prime$. Alternatively, one can numerically choose the bandwidth that minimizes the sum of asymptotic approximations to the three terms above. See Section \ref{sec: simulation} for a simulation comparison of these bandwidth selectors.}

\paragraph{Fuzzy RD Design.}
\textcolor{black}{
It is straightforward to extend our framework to the fuzzy RD design. In the fuzzy case, the target parameter is given by $\tau_{\mathtt{FRD}}\coloneqq \mathbb{E}[Y_i(1) - Y_i(0) | X_i = 0, T_i(1)> T_i(0)]$, where $T(0)\in\{0,1\}$ and $T(1)\in\{0,1\}$ denote the potential treatment indicator. We observe $T_i = T_i(1)\mathbf{1}\{X_i \geq 0\} + T_i(0)\mathbf{1}\{X_i < 0\}$. The local external validity is then assessed by the fuzzy TED, $\tau_{\mathtt{FRD}}^\prime$.
Under Assumption A3 of \cite{Dong_Lewbel:2015}, these are identified as
\begin{align*}
\tau_{\mathtt{FRD}}
= \frac{\tau_Y}{\tau_T}
\coloneqq \frac{\mu_{Y,+}-\mu_{Y,-}}{\mu_{T,+}-\mu_{T,-}},\quad
\tau_{\mathtt{FRD}}^\prime
= \frac{\tau_Y^\prime\tau_T-\tau_Y\tau_T^\prime}{\tau_T^{2}},
\end{align*}
where $\mu_{A,+}\coloneqq \lim_{x\downarrow0}\mu_{A}(x)$, $\mu_{A,-}\coloneqq \lim_{x\uparrow0}\mu_{A}(x)$, and $\mu_A(x)\coloneqq \E{A_i | X_i=x}$, where $A$ equals either $Y$ or $T$ \citep[Theorem 2]{Dong_Lewbel:2015}. The Online Appendix derives the joint asymptotic distribution of $(\tau_{\mathtt{FRD}}, \tau_{\mathtt{FRD}}^\prime)$, and hence one can draw a similar joint confidence region. For extrapolation exercise, we may assume the local linearity of $\mathbb{E}[Y_i(1) - Y_i(0) | X_i = x, T_i(1)> T_i(0)]$ in $x\in[-\delta_1,\delta_2]$. Then we can obtain analogous results to those in Lemma \ref{lemma: identification} and Proposition \ref{prop: band}. 
}

\section{Empirical Illustrations}\label{sec: empirical}
We illustrate the potential usefulness of our proposal through two emprical applications.

\subsection{Londoño-Vélez, Rodríguez, and Sánchez (2020)}
\paragraph{Empirical Context.}
\cite{Londono-Velez_etal:2020aejep} examines the effect of Ser Pilo Paga (SPP), a financial aid program introduced in Colombia in 2014, on students' probability of enrolling in higher education, using a standard RD design. 

Among our subsample of interest, consisting of students from metropolitan areas who meet a minimum academic standard ($n=19{,}738$), eligibility for the SPP scholarship is determined by a family wealth index. 
After normalizing the running variable so that the eligibility cutoff equals zero, and multiplying it by $-1$ to align with the canonical RD convention in which units above the cutoff are treated, the wealth index ranges from $-25.94$ to $56.86$, with the first quartile at $-3.76$ and the third quartile at $22.78$. Larger values correspond to poorer households, and students with an index above zero are eligible for the SPP scholarship.

\paragraph{External Validity and Extrapolation Analysis.}
We begin by assessing the standard RD effect and its local external validity, $(\tau_{\mathtt{SRD}}, \tau_{\mathtt{SRD}}^\prime)$, using the joint robust confidence region. 
For bandwidth selection and standard error estimation, we follow the \texttt{rdrobust} package \citep{rdrobust, rdrobust2} with resulting bandwidhts $(h,b)=(3.8,7.5)$.
Figure \ref{fig: spp1} reports the 95\% confidence region $\mathcal{R}_{0.95}$, where the horizontal axis represents the treatment effect and the vertical axis corresponds to the TED. The figure shows that the SPP has a sizable and statistically significant effect on enrollment probability. At the same time, the confidence region along the TED dimension is relatively tight and centred close to zero, indicating that $\tau_{\mathtt{SRD}}^\prime$ is small and statistically indistinguishable from zero.
\begin{figure}[t]
    \begin{center}
    \begin{subfigure}[b]{0.45\textwidth}
        \centering
        \includegraphics[width=1\linewidth]{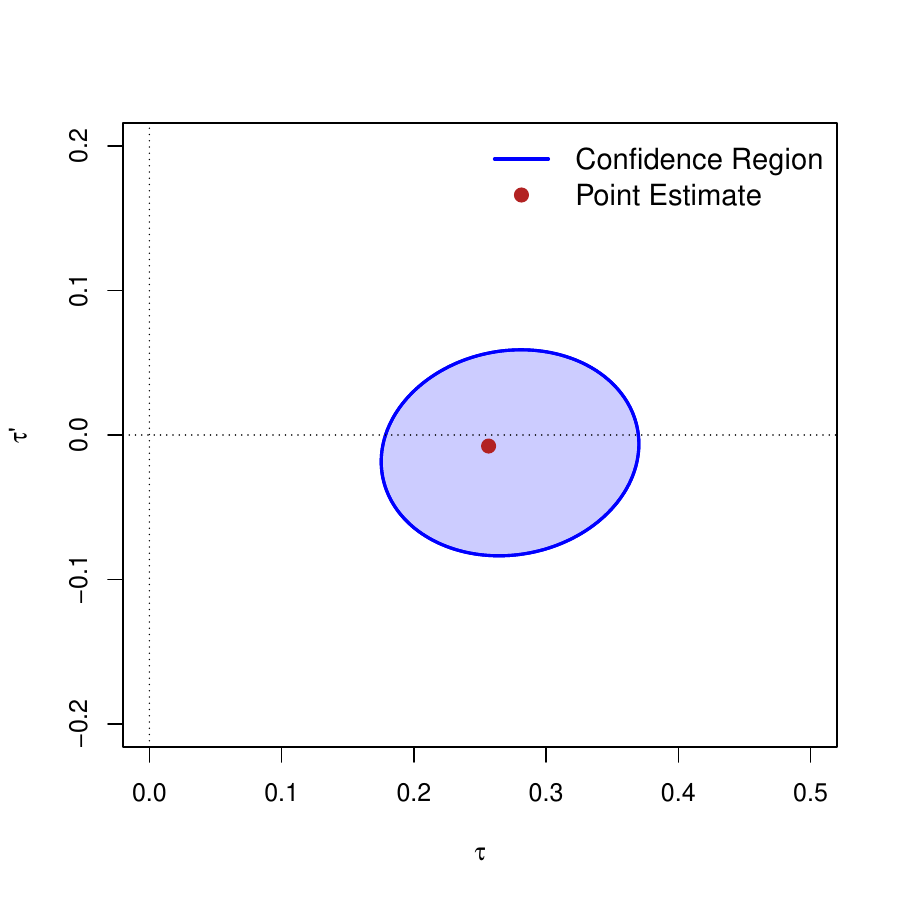}
        \caption{95\% Robust Confidence Region}
        \label{fig: spp1}
    \end{subfigure}
    \hfill
    \begin{subfigure}[b]{0.45\textwidth}
        \centering
        \includegraphics[width=1\linewidth]{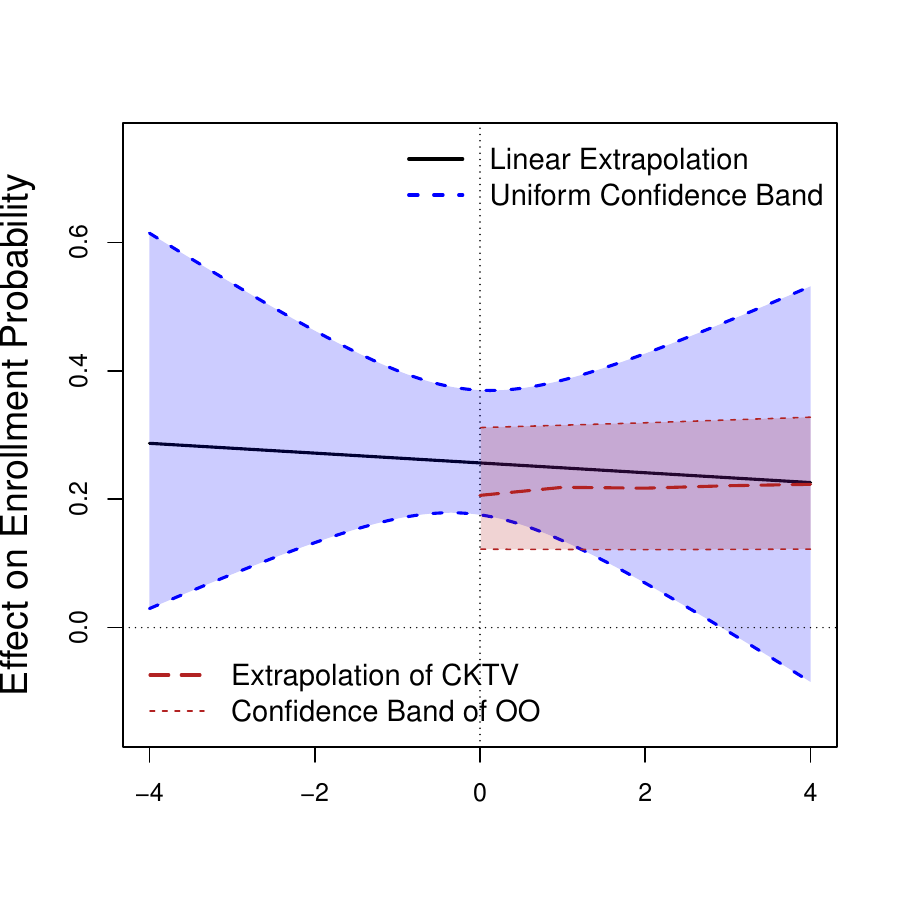}
        \caption{Linear Extrapolation and Confidence Band}
        \label{fig: spp2}
    \end{subfigure}
    \caption{External Validity and Extrapolation of Treatment Effects}
    \label{fig: spp}
    \end{center}

    \footnotesize
    \renewcommand{\baselineskip}{11pt}
    \textbf{Note:} \textbf{Panel (a).} The dependent variable is a binary indicator of whether the student enrolled in higher education. Hence, the treatment effect reported on the horizontal axis should be interpreted as the effect on the probability of enrollment.
    \textbf{Panel (b).} Linear Extrapolation indicates the point estimates $\hat{\tau}_{\mathtt{SRD}} + \hat{\tau}_{\mathtt{SRD}}^\prime x$, and Uniform Confidence Band shows $\mathcal{U}_{0.95}(x)$.
    For clarity of presentation, we report only the point estimates of \cite{Cattaneo_etal:2021JASA_extrapolating} and the uniform confidence band for the bounds proposed by \cite{Okamoto_Ozaki:2025}. Larger score values correspond to poorer households.
\end{figure}

This finding provides direct evidence of the local external validity of the RD effect of the SPP program. Importantly, such an implication could not be drawn from the treatment effect $\tau_{\mathtt{SRD}}$ alone; it emerges only when the derivative is taken into account. Hence, as emphasized by \cite{Dong_Lewbel:2015}, examining $\tau_{\mathtt{SRD}}^\prime$ alongside $\tau_{\mathtt{SRD}}$ is crucial for clarifying the policy relevance of RD estimates.
Our joint inference procedure operationalizes this insight by providing a formal statistical framework that allows both quantities to be assessed simultaneously.

While suggestive, the confidence region $\mathcal{R}_{0.95}$ in Figure \ref{fig: spp1} provides only local information and is less informative than a full extrapolation analysis. To address this limitation, we impose the locally linear treatment effects assumption over \textcolor{black}{$[-4,4]$ following the rule of thumb $\delta_1=\delta_2\approx h$}, with its plausibility briefly discussed at the end of this section. 

Figure \ref{fig: spp2} shows the extrapolated treatment effects along with the uniform confidence band $\mathcal{U}_{0.95}(x)$, which delivers clearer implications. Consistent with the robust confidence region in Figure \ref{fig: spp1}, the uniform confidence band indicates positive treatment effects over $[-4,3]$, although the effectiveness of SPP at more distant scores remains uncertain. Furthermore, the extrapolated line is rather flat, suggesting the external validity of the local estiamte $\tau(0)$.

\textcolor{black}{
The positive, seemingly constant effect over $[-4,3]$ can be economically meaningful. In fact, the extrapolated region $[-4,3]$ lies around the mode of the running variable's distribution and accounts for more than 16\% of our sample. From this perspective, the RD effect is no longer purely local; rather, it is fairly representative---at least for the subpopulation near the mode of the distribution.
}

\textcolor{black}{
To further illustrate the policy implications, consider a simple counterfactual in which the policymaker adjusts the eligibility cutoff for the SPP scholarship. The discussion below assumes that the economy is invariant to the policy change (see \citealp{Dong_Lewbel:2015} for details). Suppose the government is considering relaxing the eligibility cutoff to $c^\star=-4$. This change would expand the share of students offered the SPP scholarship by approximately 10 percentage points (from the current 66\% to about 76\% of urban students). Because the estimated treatment effect remains close to the RD effect at the original cutoff even away from the threshold, we project that the resulting increase in higher-education enrollment is approximately $0.10\times\hat{\tau}_{\mathtt{SRD}}\approx 2.6\%$.
}

\paragraph{Testing Local Linearity.}
\textcolor{black}{
The key assumption of the previous extrapolation analysis is the local linearity of the treatment effects function $\tau(x)$. As a falsification test, one can test $\tau^{(2)}(0)=0$. We employ the local cubic smoother and the second derivative is estiamted as $\hat{\tau}^{(2)}(0)=0.009$ with $p$-value of $0.891$ \citep{calonico2020optimal}.
}

\textcolor{black}{
In addition, the linearity also consistent with the findings from alternative approaches that exploit multiple cutoff points, such as \cite{Cattaneo_etal:2021JASA_extrapolating} and \cite{Okamoto_Ozaki:2025}, represented by the red dashed lines and shaded region in Figure \ref{fig: spp2}.\footnote{In this application, there exists another group of students from rural regions who face a higher (i.e., more strict) cutoff point. As a result, the regression function of the enrollment probability on the wealth index under the control status can be observed for this additional group even on the right-hand side of the cutoff. \citeauthor{Cattaneo_etal:2021JASA_extrapolating}’s (\citeyear{Cattaneo_etal:2021JASA_extrapolating}) strategy achieves point-identified extrapolation by assuming that the regression functions of rural and metropolitan students are ``parallel" over $[0,4]$. By contrast, \cite{Okamoto_Ozaki:2025} assumes that the regression function for rural students is monotonically decreasing as the household becomes poorer, and that it lies below that of metropolitan students, thereby yielding partial identification.}
The linearly extrapolated line on the right-hand side of the cutoff lies within the confidence band of \cite{Okamoto_Ozaki:2025} and closely resembles the extrapolated curve of \cite{Cattaneo_etal:2021JASA_extrapolating}.
This concordance supports the plausibility of the locally linear treatment effects assumption in this application.
}

\subsection{Holden (2016)}
\paragraph{Empirical Context.}
\textcolor{black}{
\cite{Holden:2016} examines the effect of a one-time textbook funding intervention stemming from the 2004 Williams settlement in California. 
Elementary schools in the bottom two deciles of the 2003 Academic Performance Index (API)---elementary schools with API $\leq$ 643---received an extra \$96.90 per student on top of a base \$54.22.
Exploiting this discontinuity, the paper implements a sharp RD design to estimate the local causal effect of textbook funding on school-level standardized test performance, finding a positive effect.
}

\textcolor{black}{
We again normalize the API by centering it at the threshold (cutoff $=$ 0) and multiplying by $-1$ so that values above zero correspond to treatment; in this scale, the running variable ranges from $-353$ to $268$.
Larger values correspond to low performing schools.
}

\paragraph{External Validity and Extrapolation Analysis.}
\textcolor{black}{
Figure \ref{fig: book1} reports the $95\%$ joint confidence region, using the bandwidths $(h,b)=(20.5,42.3)$ estimated by \texttt{rdrobust} package, to assess the external validity of the RD estimate.
Although the estimated treatment effect is positive and statistically significant, we find that the estimated TED is negative and is also statistically differ from zero.
This may suggest that the treatment effect function is not a constant function even locally, thereby casting a certain concern about the external validity of the standard RD effect.
}
\begin{figure}[t]
    \begin{center}
    \begin{subfigure}[b]{0.45\textwidth}
        \centering
        \includegraphics[width=1\linewidth]{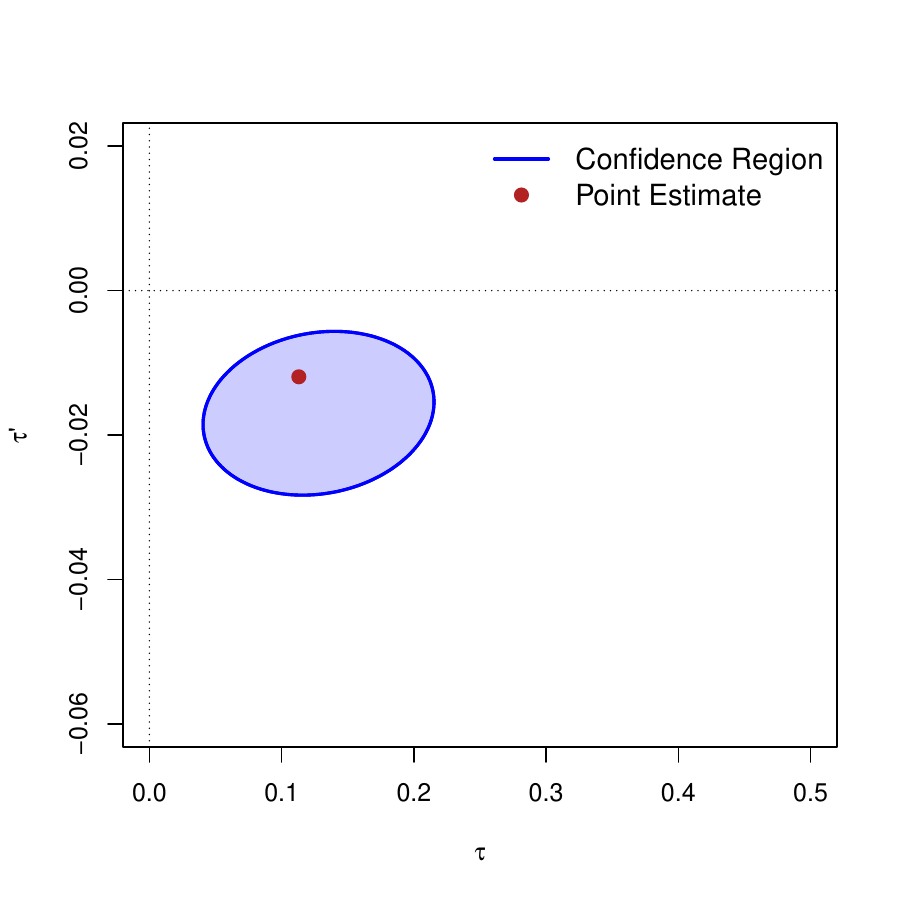}
        \caption{95\% Robust Confidence Region}
        \label{fig: book1}
    \end{subfigure}
    \hfill
    \begin{subfigure}[b]{0.45\textwidth}
        \centering
        \includegraphics[width=1\linewidth]{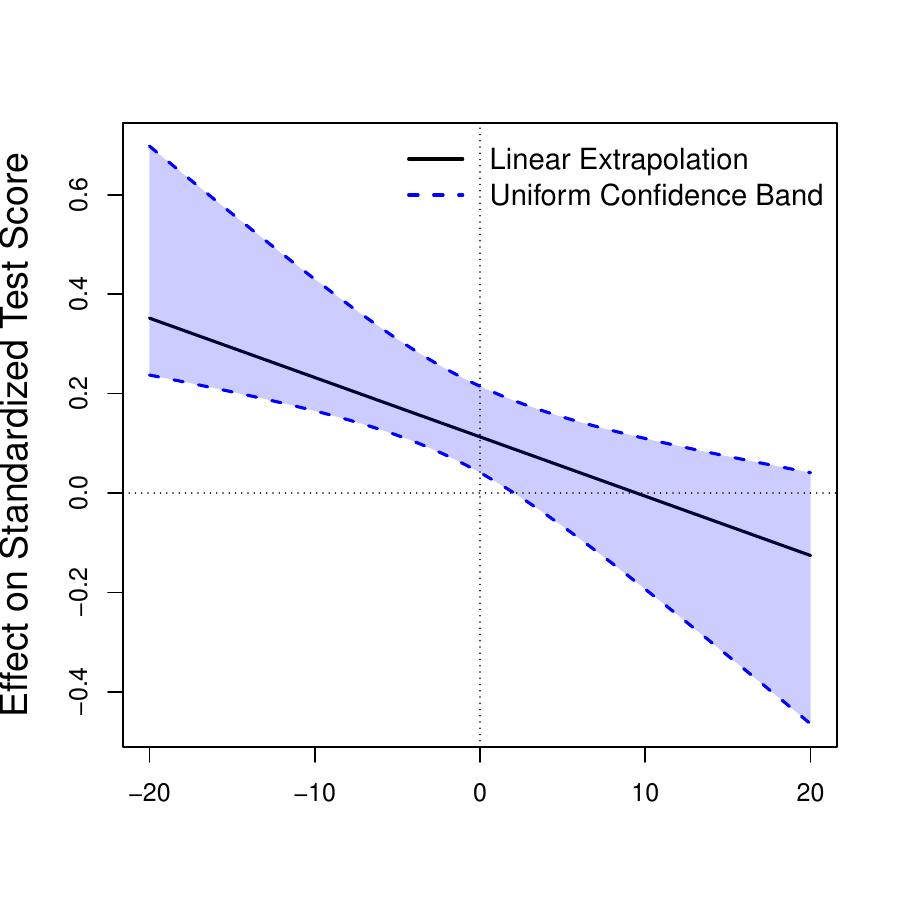}
        \caption{Linear Extrapolation and Confidence Band}
        \label{fig: book2}
    \end{subfigure}
    \caption{External Validity and Extrapolation of Treatment Effects}
    \label{fig: book}
    \end{center}

    \footnotesize
    \renewcommand{\baselineskip}{11pt}
    \textbf{Note:} \textbf{Panel (a).} The dependent variable is a normalized test score.
    \textbf{Panel (b).} Linear Extrapolation indicates the point estimates $\hat{\tau}_{\mathtt{SRD}} + \hat{\tau}_{\mathtt{SRD}}^\prime x$, and Uniform Confidence Band shows $\mathcal{U}_{0.95}(x)$. Larger score values correspond to low performing schools.
\end{figure}

\textcolor{black}{
To more clearly assess the external validity of the RD estimate, we impose a locally linear treatment effects assumption over the window $[-20,20]$, following the rule of thumb choice $\delta_1=\delta_2\approx h$.
We compute the implied treatment effects over the region in Figure \ref{fig: book2}. 
The figure yields two salient implications. First, treatment effects can become statistically indistinguishable from zero even for slightly larger score values. Second, effects appear comparatively larger at lower score values within the extrapolation window.
}

\textcolor{black}{
These patterns suggest that the effect identified at the cutoff (API $=643$) may not generalize uniformly across the score distribution. In particular, the local RD estimate could overstate the effect for lower-performing schools. 
One interpretation is that schools struggling with low educational attainment may benefit from more direct and intensive forms of support, such as an after-school program \citep{Muralidharan:2019}, and that textbook funding alone may be insufficient to generate sizable improvements, consistent with \cite{Glewwe:2009}.
}

\textcolor{black}{
Conversely, the extrapolation also hints that better-achieving schools could experience effects that exceed the RD estimate at the cutoff. This pattern would be consistent with complementarities between improved access to instructional materials and pre-existing student ability or school capacity. 
For instance, higher-performing students may better benefit from additional textbooks \citep{Glewwe:2009}. Likewise, schools with greater instructional capacity (e.g., more experienced teachers) may be better able to integrate new materials into classroom practice, thereby translating improved textbook access into larger test-score gains, consistent with \cite{Mueller:2013}.
}

\textcolor{black}{
Building on these findings, we can also consider a simple, illustrative policy experiment: if treatment effects are indeed heterogeneous across the running variable, then reallocating a fixed textbook budget could improve cost-effectiveness. 
In our sample, $\P{-20<X_i\leq0}\approx\P{0\leq X_i<20}\approx 6\%$, so excluding the interval $[0,20)$ while covering $[-20,0]\cup[20,\infty)$\footnote{\textcolor{black}{Treatment effects over the region $[20,\infty)$ are uncertain in our framework, so we leave the treatment assignment rule for those schools unchanged.}} would keep the number of eligible schools roughly comparable to the status quo but shift resources toward score ranges where the extrapolated effects are not attenuated. Under this counterfactual assignment rule, the textbook funding program may yield larger average test-score gains.
}

\paragraph{Testing Local Linearity.}
\textcolor{black}{
In the absense of the multi-cutoffs as in the previous illustration, we have to rely on the test of $\tau^{(2)}(0)=0$. We employ the local cubic smoother and the second derivative is estiamted as $\hat{\tau}^{(2)}(0)=-0.006$ with $p$-value of $0.056$.
}

\section{Simulation Studies}\label{sec: simulation}
We present numerical experiments that assess the finite-sample performance of the proposed procedure.
We consider the following three data-generating processes (DGPs): $Y_i = \mu_-(X_i)\mathbf{1}\{X_i<0\}+\mu_+(X_i)\mathbf{1}\{X_i\geq0\} + \varepsilon_i$, $X_i \overset{\mathrm{iid}}{\sim} \mathrm{Uniform}[-5,5]$, and $\varepsilon_i\overset{\mathrm{iid}}{\sim}\mathcal{N}(0,1)$, where
\begin{align}
    &\mu_-(x) = x + \eta(x),\quad
    \mu_+(x) = 2 + x + \eta(x),
    \tag{DGP 1}\label{eq: dgp1}\\
    &\mu_-(x) = x + \eta(x),\quad
    \mu_+(x) = 2 - x + \eta(x),
    \tag{DGP 2}\label{eq: dgp2}\\
    &\mu_-(x) = x - \frac{1}{5}x^2 + \eta(x),\quad
    \mu_+(x) = 2 - x - \frac{1}{10} x^2 + \eta(x)
    \tag{DGP 3}\label{eq: dgp3},
\end{align}
and $\eta(x) \coloneqq \sin(x)/10 + \cos(x)/10$. Plots of the DGPs are provided in the Online Appendix.
In every case, the standard RD effect is $\tau_{\mathtt{SRD}} = 2$.
In contrast the TED is different: $\tau_{\mathtt{SRD}}^\prime = 0$ for \ref{eq: dgp1}, and $\tau_{\mathtt{SRD}}^\prime = -2$ for \ref{eq: dgp2} and \ref{eq: dgp3}. For the first two DGPs, the treatment effect function is linear in $x$, while it is not for \ref{eq: dgp3}. 

\paragraph{One-Shot Simulation Illustration.}
\begin{figure}[t]
    \begin{center}
    \begin{subfigure}[b]{0.32\textwidth}
        \centering
        \includegraphics[width=\linewidth]{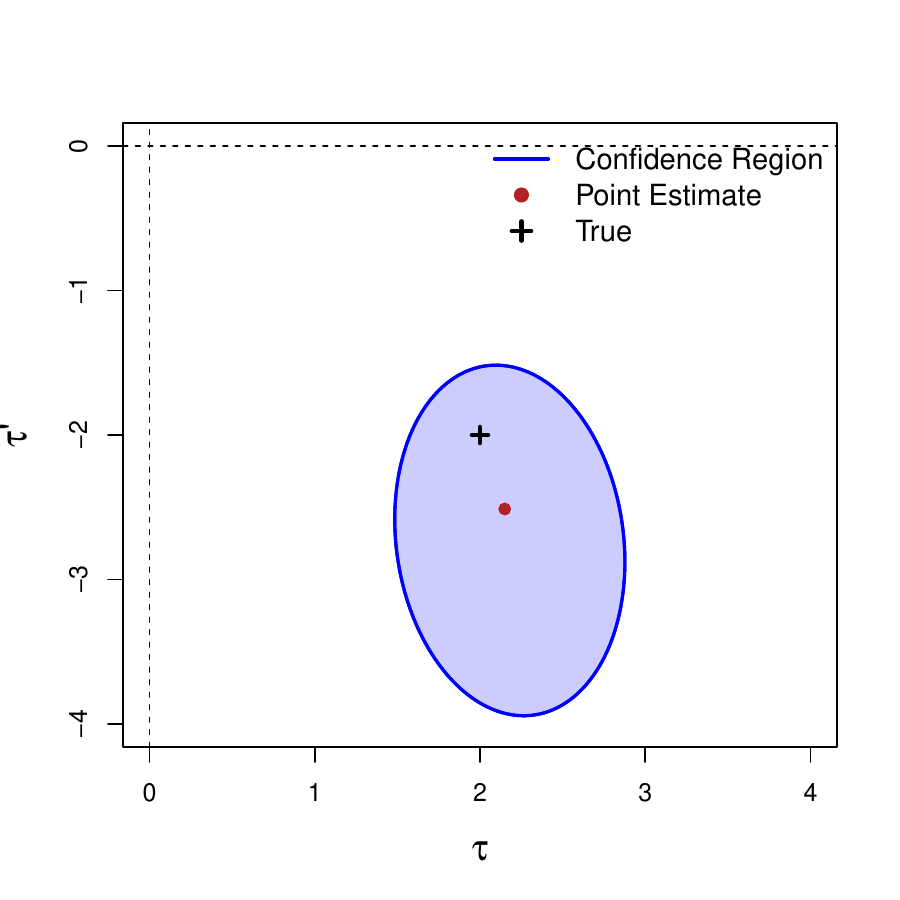}
        \caption{$\mathcal{R}_{0.95}$}
        \label{fig: simu1}
    \end{subfigure}
    \hfill
    \begin{subfigure}[b]{0.32\textwidth}
        \centering
        \includegraphics[width=\linewidth]{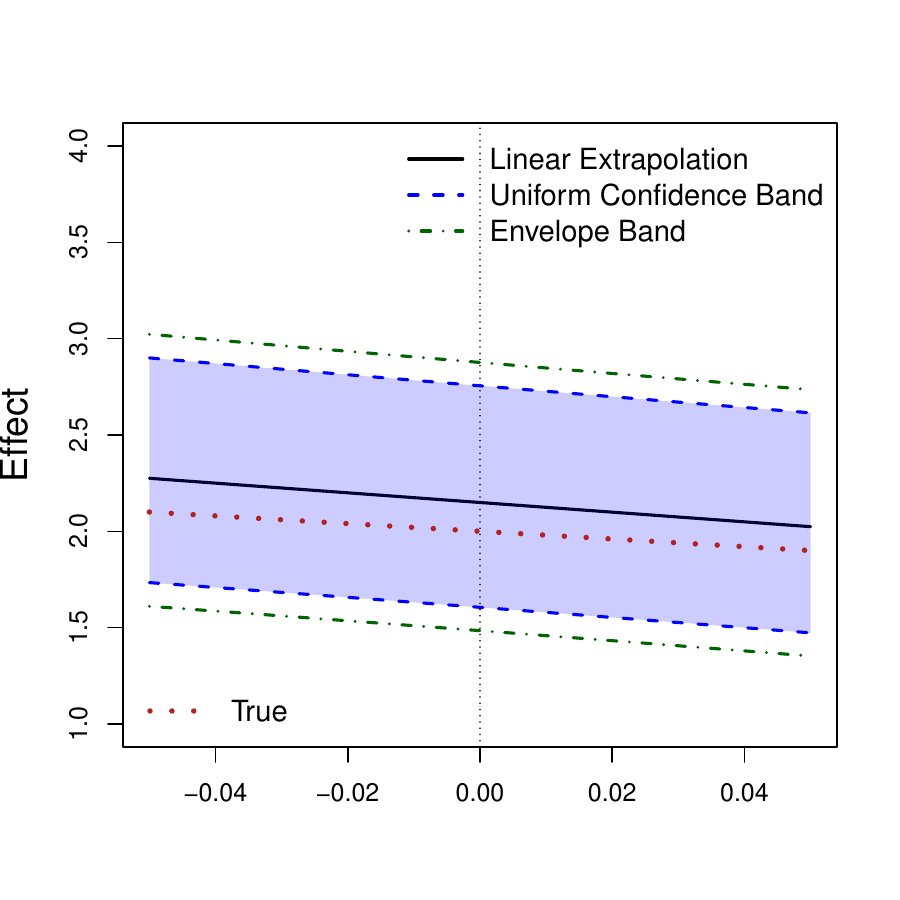}
        \caption{$\mathcal{U}_{0.95}$ ($\delta_1=\delta_2=0.05$)}
        \label{fig: simu2}
    \end{subfigure}
    \hfill
    \begin{subfigure}[b]{0.32\textwidth}
        \centering
        \includegraphics[width=\linewidth]{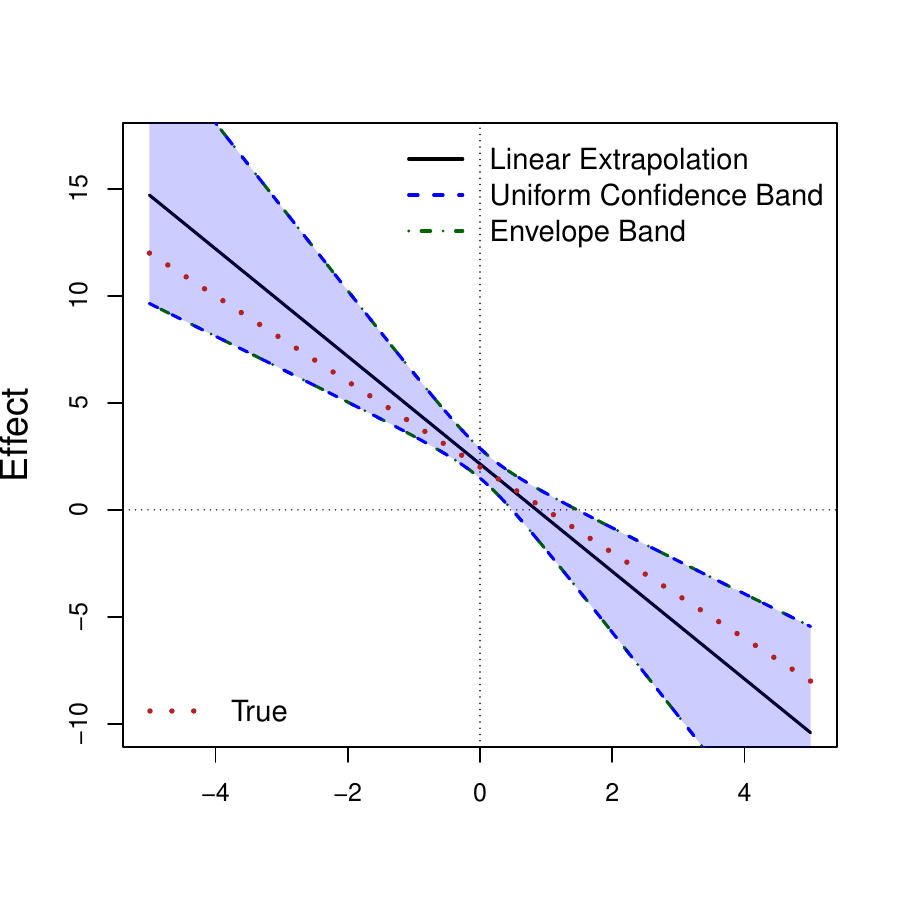}
        \caption{$\mathcal{U}_{0.95}$ ($\delta_1=\delta_2=5$)}
        \label{fig: simu3}
    \end{subfigure}

    \caption{Joint Confidence Region \& Confidence Band \eqref{eq: dgp2}}
    \label{fig: simulation}
    \end{center}

    \footnotesize
    \renewcommand{\baselineskip}{11pt}
    \textbf{Note:} Point Estimate and Linear Extrapolations are reported before bias correction, i.e., $(\hat{\tau}_{\mathtt{SRD}} , \hat{\tau}_{\mathtt{SRD}}^\prime)$ and $\hat{\tau}_{\mathtt{SRD}} + \hat{\tau}_{\mathtt{SRD}}^\prime x$.
\end{figure}
As an illustration, Figure \ref{fig: simulation} presents a single-shot simulation result of \ref{eq: dgp2} with a sample size of $1000$.
Panel (a) shows the confidence region $\mathcal{R}_{0.95}$, which correctly covers the true set $(\tau_{\mathtt{SRD}}, \tau_{\mathtt{SRD}}^\prime) = (2,-2)$. From this picture, we can deduce that the treatment effect at the cutoff is positive, while it decreases as the running variable $X$ increases.
Panels (b) and (c) display the uniform confidence band under the locally linear treatment effects assumption over $[-0.05, 0.05]$ and $[-5, 5]$, respectively. In both cases, the uniform confidence band contains the true treatment effect function indicated by the red dotted line. 
We also observe that the confidence band $\mathcal{U}_{0.95}$ converges to the envelope band $\mathcal{E}_{0.95}$ as $(\delta_1,\delta_2)$ increase.

\paragraph{Monte Carlo Results.}
\textcolor{black}{
We conduct a Monte Carlo experiment. The main results are reported for $n=1000$; results for $n=500$ are deferred to the Online Appendix and lead to the same qualitative conclusions. 
We consider three symmetric extrapolation regions, $\delta_1=\delta_2=\delta\in\{0.05, 1, 5\}$, ranging from a very local window to a wide extrapolation window.
Throughout, we consider local linear estimation ($p=1$) as the baseline, and we additionally report results for local quadratic estimation ($p=2$) to assess whether using a higher-order polynomial improves finite-sample performance in this setting.}

\textcolor{black}{
We compare several bandwidth selectors that correspond to the alternative criteria discussed in Section \ref{sec: remarks}.
First, we use the plug-in MSE-optimal bandwidth for the RD level effect, denoted by $\hat{h}_{\mathtt{mean}}$, which minimizes the estimated MSE of $\hat{\tau}_{\mathtt{SRD}}$ based on the first-order asymptotic approximation. 
Second, we use the plug-in MSE-optimal bandwidth for the TED, denoted by $\hat{h}_{\mathtt{deriv}}$, which minimizes the estimated MSE of $\hat{\tau}_{\mathtt{SRD}}^\prime$. Although this is the asymptotically MISE optimal choice for the linear extrapolation, we also consider an alternative MISE-oriented choice, $\hat{h}_{\mathtt{both}}$, obtained by minimizing $\widehat{\mathrm{MISE}}
= \widehat{\mathrm{MSE}}\left[\hat{\tau}_{\mathtt{SRD}}\right]\int_{-\delta}^{\delta} 1\,dx
+ \widehat{\mathrm{MSE}}\left[\hat{\tau}^\prime_{\mathtt{SRD}}\right]\int_{-\delta}^{\delta} x^2\,dx$, which may perform better in finite samples.
Finally, we report results using the regularized MSE-optimal bandwidths for $\hat{\tau}_{\mathtt{SRD}}$ and $\hat{\tau}_{\mathtt{SRD}}^\prime$, respectively denoted by $\hat{h}_{\mathtt{reg0}}$ and $\hat{h}_{\mathtt{reg1}}$, following \citet{Imbens_Kaluanaraman:2011} as implemented in the \texttt{rdrobust} package \citep{rdrobust2}. For bias correction, we use the same MSE-optimal bandwidth for bias estimation as that used by \texttt{rdrobust} throughout, regardless of the bandwidth selector used for point estimation.
}

\textcolor{black}{
We first discuss the results for local linear regression ($p=1$).
The left panel of Table \ref{tab: simulation} reports the empirical coverage probabilities of the proposed confidence region and band, as well as the MISE of the linearly extrapolated effects, based on $1{,}000$ Monte Carlo replications.
In all cases, the coverage probability of the joint confidence region $\mathcal{R}_{0.95}$ is excellent and close to the nominal level.
The coverage properties are insensitive to the bandwidth choice, reflecting the \textit{robust} nature of the RBC framework. Plus, no single bandwidth selector appears to uniformly dominate the others in terms of empirical coverage probabilities.
When the (locally) linear treatment effects assumption holds (\ref{eq: dgp1} and \ref{eq: dgp2}), the uniform confidence band $\mathcal{U}_{0.95}(x)$ also performs well and exhibits a similar coverage accuracy.
However, in \ref{eq: dgp3}, as this assumption is violated, $\mathcal{U}_{0.95}(x)$ exhibits undercoverage when the extrapolation interval is wide (i.e., when $\delta=5$). 
By contrast, when the extrapolation interval is relatively narrow (i.e., $\delta=1$ and $0.05$), the linearity assumption provides a reasonable approximation to the true treatment effects function, resulting in satisfactory coverage.
}

\textcolor{black}{
In terms of MISE, bandwidths that target the MISE criterion ($\hat{h}_{\mathtt{deriv}}$ and $\hat{h}_{\mathtt{both}}$) tend to outperform $\hat{h}_{\mathtt{mean}}$, especially for larger extrapolation windows (i.e., $\delta=1$ and $5$). 
Nevertheless, these differences are modest when translated into the uniformly averaged pointwise root MSE, $\sqrt{\mathrm{MISE}/(2\delta)}$ (ARMSE), and the RMSE at $x=\pm\delta$ (MSE is reported in the Online Appendix). For $\delta=1$ under \ref{eq: dgp1} and \ref{eq: dgp2}, the ARMSE gap between $\hat{h}_{\mathtt{mean}}$ and $\hat{h}_{\mathtt{both}}$ is only about $0.07$. 
For RMSE at $x=\pm 1$, the gap is $0.09$--$0.12$. 
Even for the larger window $\delta=5$, the difference is about $0.29$--$0.32$ in ARMSE and $0.48$--$0.56$ in RMSE, which does not appear substantial relative to the scale of the estimand (see Figure~\ref{fig: simu3}). 
The regularized bandwidths, $\hat{h}_{\mathtt{reg0}}$ and $\hat{h}_{\mathtt{reg1}}$, improve upon their non-regularized counterparts, perhaps due to the stabilizing effect of regularization. Interestingly, $\hat{h}_{\mathtt{reg0}}$ performs best in terms of finite-sample MISE in our simulations. One contributing factor is that $\hat{h}_{\mathtt{reg0}}$ yields a smaller $\mathrm{MSE}[\hat{\tau}_{\mathtt{SRD}}^\prime]$ than $\hat{h}_{\mathtt{reg1}}$ in our setup. 
Even so, the resulting differences in ARMSE and RMSE between the two remain limited (e.g., the ARMSE gap is around $0.05$ for $\delta=1$).
Overall, both coverage and MISE appear relatively insensitive to the choice of bandwidth selector. That said, given its consistently better MISE performance, we recommend using the usual \texttt{rdrobust} bandwidth $\hat{h}_{\mathtt{reg0}}$ as the default choice.
}

\textcolor{black}{
The right panel of Table \ref{tab: simulation} shows the simulation results for the local quadratic estimation ($p=2$).
Although main implications remain, the coverage property and MISE both become worse compared to the local linear cases. Hence, except for the case when the sample size is extremely large, the local linear specification would be a better option.
}
\begin{landscape}
\begin{table}
    \centering
    \begin{tabular}{lccccccccccc}\hline\hline
         & \multicolumn{5}{c}{Local Linear ($p=1$)} & & \multicolumn{5}{c}{Local Quadratic ($p=2$)}\\\cline{2-6}\cline{8-12}
         & $\hat{h}_{\mathtt{mean}}$ & $\hat{h}_{\mathtt{deriv}}$ & $\hat{h}_{\mathtt{both}}$ & $\hat{h}_{\mathtt{reg0}}$ & $\hat{h}_{\mathtt{reg1}}$ & & $\hat{h}_{\mathtt{mean}}$ & $\hat{h}_{\mathtt{deriv}}$ & $\hat{h}_{\mathtt{both}}$ & $\hat{h}_{\mathtt{reg0}}$ & $\hat{h}_{\mathtt{reg1}}$\\\hline
        \multicolumn{4}{l}{\ref{eq: dgp1}: Coverage}\\
        \hspace{0.5cm}$\mathcal{R}_{0.95}$ & 0.943 & 0.944 & 0.942 & 0.936 & 0.945 & & 0.926 & 0.928 & 0.927 & 0.931 & 0.934 \\
        \hspace{0.5cm}$\mathcal{U}_{0.95}$ ($\delta=0.05$) & 0.940 & 0.944 & 0.941 & 0.936 & 0.945 & & 0.928 & 0.929 & 0.927 & 0.924 & 0.926 \\
        \hspace{0.5cm}$\mathcal{U}_{0.95}$ ($\delta=1$) & 0.942 & 0.942 & 0.939 & 0.935 & 0.945 & & 0.927 & 0.928 & 0.924 & 0.930 & 0.935 \\
        \hspace{0.5cm}$\mathcal{U}_{0.95}$ ($\delta=5$) & 0.943 & 0.944 & 0.945 & 0.936 & 0.945 & & 0.926 & 0.928 & 0.927 & 0.931 & 0.934 \\
        \multicolumn{4}{l}{\ref{eq: dgp1}: MISE}\\
        \hspace{0.5cm}$\delta=0.05$ & 0.013 & 0.016 & 0.013 & 0.010 & 0.011 & & 0.047 & 0.056 & 0.046 & 0.017 & 0.017 \\
        \hspace{0.5cm}$\delta=1$ & 0.945 & 0.868 & 0.754 & 0.641 & 0.756 & & 3.725 & 3.160 & 2.623 & 1.855 & 2.108 \\
        \hspace{0.5cm}$\delta=5$ & 85.93 & 70.44 & 67.97 & 55.59 & 68.69 & & 353.6 & 260.1 & 249.6 & 191.6 & 222.6 \\\hline
        
        \multicolumn{4}{l}{\ref{eq: dgp2}: Coverage}\\
        \hspace{0.5cm}$\mathcal{R}_{0.95}$ & 0.944 & 0.942 & 0.943 & 0.940 & 0.940 & & 0.933 & 0.938 & 0.932 & 0.931 & 0.929 \\
        \hspace{0.5cm}$\mathcal{U}_{0.95}$ ($\delta=0.05$) & 0.937 & 0.940 & 0.937 & 0.934 & 0.931 & & 0.927 & 0.935 & 0.926 & 0.930 & 0.925 \\
        \hspace{0.5cm}$\mathcal{U}_{0.95}$ ($\delta=1$) & 0.944 & 0.939 & 0.939 & 0.939 & 0.938 & & 0.932 & 0.936 & 0.936 & 0.927 & 0.927 \\
        \hspace{0.5cm}$\mathcal{U}_{0.95}$ ($\delta=5$) & 0.944 & 0.942 & 0.945 & 0.940 & 0.940 & & 0.933 & 0.938 & 0.937 & 0.931 & 0.929 \\
        \multicolumn{4}{l}{\ref{eq: dgp2}: MISE}\\
        \hspace{0.5cm}$\delta=0.05$ & 0.013 & 0.015 & 0.013 & 0.010 & 0.011 & & 0.037 & 0.051 & 0.036 & 0.016 & 0.017 \\
        \hspace{0.5cm}$\delta=1$ & 0.883 & 0.815 & 0.718 & 0.637 & 0.745 & & 2.999 & 2.840 & 2.222 & 1.679 & 1.862 \\
        \hspace{0.5cm}$\delta=5$ & 79.75 & 67.16 & 64.35 & 55.49 & 67.65 & & 287.2 & 232.5 & 221.5 & 170.9 & 193.1 \\\hline
        
        \multicolumn{4}{l}{\ref{eq: dgp3}: Coverage}\\
        \hspace{0.5cm}$\mathcal{R}_{0.95}$ & 0.937 & 0.945 & 0.937 & 0.942 & 0.946 & & 0.929 & 0.925 & 0.925 & 0.923 & 0.927 \\
        \hspace{0.5cm}$\mathcal{U}_{0.95}$ ($\delta=0.05$) & 0.942 & 0.940 & 0.942 & 0.943 & 0.941 & & 0.925 & 0.925 & 0.923 & 0.928 & 0.926 \\
        \hspace{0.5cm}$\mathcal{U}_{0.95}$ ($\delta=1$) & 0.939 & 0.940 & 0.936 & 0.938 & 0.945 & & 0.928 & 0.924 & 0.926 & 0.925 & 0.928 \\
        \hspace{0.5cm}$\mathcal{U}_{0.95}$ ($\delta=5$) & 0.890 & 0.884 & 0.881 & 0.868 & 0.904 & & 0.918 & 0.907 & 0.907 & 0.912 & 0.914 \\
        \multicolumn{4}{l}{\ref{eq: dgp3}: MISE}\\
        \hspace{0.5cm}$\delta=0.05$ & 0.014 & 0.011 & 0.014 & 0.010 & 0.011 & & 0.045 & 0.050 & 0.044 & 0.017 & 0.017 \\
        \hspace{0.5cm}$\delta=1$ & 0.980 & 0.742 & 0.673 & 0.635 & 0.764 & & 3.367 & 3.057 & 2.430 & 1.829 & 2.037 \\
        \hspace{0.5cm}$\delta=5$ & 101.5 & 77.82 & 75.55 & 68.33 & 82.31 & & 326.4 & 274.9 & 255.5 & 201.3 & 226.7 \\\hline
    \end{tabular}
    \caption{Simulation Results ($n=1000$)}
    \label{tab: simulation}
\end{table}
\end{landscape}

\section{Conclusion and Recommendations}\label{sec:conclusion}
Motivated by our survey suggesting that recent empirical RD studies pay little attention to external validity and that more accessible econometric tools are needed, this article proposes a simple procedure for jointly inferring the RD effect and its local external validity. Building on the insight of \cite{Calonico_etal:2014} and \cite{Dong_Lewbel:2015}, we proposed an RBC-based confidence region for the treatment effect and TED at the cutoff point. We then introduced a locally linear treatment effects assumption, which complements and extends the framework of \cite{Dong_Lewbel:2015} by allowing for more direct extrapolation analysis of RD effects. Under this assumption, we further derived a uniform confidence band for the linearly extrapolated effects.

\textcolor{black}{
In any RD setting, we recommend reporting both the standard RD effect at the cutoff, $\tau_{\mathtt{SRD}}$, and the treatment-effect derivative (TED), $\tau_{\mathtt{SRD}}^\prime$, of \citet{Dong_Lewbel:2015}.
Under mild smoothness conditions, the TED serves as a broadly applicable diagnostic for the local external validity of $\tau_{\mathtt{SRD}}$, and thus can be reported in virtually all RD studies.
Our confidence region $\mathcal{R}_{1-\alpha}$ provides a unified way to make this assessment with formal uncertainty quantification.
Moreover, our numerical results support the use of conventional bandwidth choices, so no additional tuning parameters or design modifications are required for valid inference.}

\textcolor{black}{
When treatment effects away from the cutoff are of substantive or policy interest, the choice of an extrapolation strategy should reflect the available design features and the credibility of any additional assumptions. In the absence of design features---as in our second application---our linear extrapolation approach is appealing because it requires nothing beyond the standard RD setup.
Even when multiple cutoffs are available, if the research interest lies to the left of the cutoff, our procedure can still be useful because the multi-cutoff extrapolation methods proposed in the literature are not applicable in that direction \citep{Cattaneo_etal:2021JASA_extrapolating, Okamoto_Ozaki:2025}. For example, the policy experiment in our first application (relaxing the eligibility threshold) requires extrapolation to the left.
By contrast, when the region to the right of the cutoff is of interest, all three approaches are, at least mechanically, applicable. Since each approach relies on certain functional-form assumptions, researchers should choose among them based on the plausibility of the corresponding identifying assumptions.
Relatedly, with many cutoffs, detailed counterfactual policy effects (in terms of treatment-assignment rules) can be studied using the methodology developed by \citet{Bertanha:2020}.
Beyond these multi-cutoff approaches, informative covariates can also help extrapolate treatment effects \citep{Angrist_Rokkanen:2015jasa, Bertanha_Imbens:2020jbes, Deaner_Kwon:2025}.}

\textcolor{black}{
Overall, we view joint estimation and inference for $(\tau_{\mathtt{SRD}},\tau_{\mathtt{SRD}}^\prime)$ as a natural default.
When extrapolated effects themselves are the target, one can then layer on additional design features or stronger identifying assumptions to draw more informative conclusions.}

\newpage
\setcounter{section}{0}
\setcounter{page}{1}
\renewcommand{\thepage}{S\arabic{page}}
\setcounter{equation}{0}
\renewcommand{\theequation}{S.\arabic{equation}}
\setcounter{lemmax}{0}
\renewcommand{\thelemmax}{S\arabic{lemmax}}
\setcounter{theoremx}{0}
\renewcommand{\thetheoremx}{S\arabic{theoremx}}
\setcounter{corx}{0}
\renewcommand{\thetheoremx}{S\arabic{corx}}
\setcounter{table}{0}
\renewcommand{\thesection}{S\arabic{section}}
\setcounter{table}{0}
\renewcommand{\thetable}{S\arabic{table}}
\setcounter{figure}{0}
\renewcommand{\thefigure}{S\arabic{figure}}
\setcounter{remarkx}{0}
\renewcommand{\theremarkx}{S\arabic{remarkx}}


\makeatletter
\renewcommand*{\@fnsymbol}[1]{\ensuremath{\ifcase#1\or \flat\or * \else\@ctrerr\fi}}
\makeatother

\begin{center}
{\Large\bf Online Appendix for ``Joint Inference for the Regression Discontinuity Effect and Its External Validity"}
\end{center}
\begin{center}
    {\large Yuta Okamoto}
\end{center}

\section{Proofs and Generalization}\label{sec: proofs}
\subsection{Proofs of Lemma 1 and Proposition 1}\label{sec:proof1}
This subsection provides proofs of Lemma 1 and Proposition 1 in a more general setting, where we employ $p(\geq1)$-th order local polynomial regression for the main estimation and use local polynomial regression of order $q(\geq p+1)$ for the bias estimation.
As the matrix notation is more convenient, we begin by introducing some additional notations.
Our $p$-th order local polynomial estimators are given by $\hat{\mu}_+ = e_0^\top \hat{\beta}_{+,p}$ and $\hat{\mu}_+^\prime = e_1^\top \hat{\beta}_{+,p}$, where $e_0 \coloneqq (1,0,\ldots,0)^\top \in \mathbb{R}^{p+1}$, $e_1 \coloneqq (0,1,0,\ldots,0)^\top \in \mathbb{R}^{p+1}$,
\begin{align*}
    \hat{\beta}_{+,p}
    \coloneqq \argmin_{\beta \in\mathbb{R}^{p+1}} \sum_{i=1}^{n}\mathbf{1}\left\{X_i \geq 0\right\} \left(Y_i - r_p(X_i)^\top \beta\right)^2 \frac{1}{h}K\left(\frac{X_i}{h}\right),
\end{align*}
and $r_p(u) \coloneqq (1, u, \ldots, u^p)^\top$.
$\hat{\mu}_{-}$ and $\hat{\mu}_{-}^\prime$ is defined similarly.
Note that $\hat{\mu}_{+}$ and $\hat{\mu}_{+}^\prime$ have the following closed form expression:
\begin{align*}
    \hat{\mu}_{+} = \frac{1}{n} e_0^\top \Gamma_+^{-1} R^\top W_+ Y,\quad
    \hat{\mu}_{+}^\prime = \frac{1}{n}\frac{1}{h} e_1^\top \Gamma_+^{-1} R^\top W_+ Y,
\end{align*}
where $Y \coloneqq (Y_1,\ldots,Y_n)^\top$, $W_+ \coloneqq \mathrm{diag}(\mathbf{1}\{X_1\geq0\}K(X_1/h)/h,\ldots, \mathbf{1}\{X_n\geq0\}K(X_n/h)/h)$, $R \coloneqq (r_p(X_1/h),\ldots,r_p(X_n/h))^\top$, and $\Gamma_+ \coloneqq R^\top W_+ R / n$. The counterparts replacing $+$ with $-$ are defined similarly.
We will focus on the $+$ case as long as the $-$ case can be treated analogously below.

We assume $n \min\{h,b\}\to\infty$, $n \min\{h^{2(p+1)+1}, b^{2(p+1)+1}\} \max\{h^2, b^{2(q-p)}\}\to0$, and $\max\{h,b\}<\kappa_0$ throughout.
We also assume Assumptions 1-2 hold with $S\geq q+1 \geq p+2$.

The usual Taylor expansion argument suggests that
\begin{align*}
    \E{\hat{\mu}_+ | X_1,\ldots,X_n} - {\mu}_+ &\approx h^{p+1}\frac{\mu_+^{(p+1)}}{(p+1)!}e_0^\top\Gamma_+^{-1}\vartheta_+ \eqqcolon h^{p+1}\frac{\mu_+^{(p+1)}}{(p+1)!} \mathcal{B}_+\\
    \E{\hat{\mu}_+^\prime | X_1,\ldots,X_n} - {\mu}_+^\prime &\approx h^{p}\frac{\mu_+^{(p+1)}}{(p+1)!}e_1^\top\Gamma_+^{-1}\vartheta_+\eqqcolon h^{p}\frac{\mu_+^{(p+1)}}{(p+1)!} \mathcal{B}_+^\prime,
\end{align*}
where $\vartheta_+ \coloneqq R^\top W_+ S/n$ and $S\coloneqq ((X_1/h)^{p+1},\ldots, (X_n/h)^{p+1})^\top$.
We use the notations $\mathcal{B}_+$ and $\mathcal{B}_+^\prime$, with a slight abuse of notation, to keep the exposition intuitive.
This motivates our bias-corrected estimators:
\begin{align*}
    \tilde{\mu}_+ \coloneqq \hat{\mu}_+ - h^{p+1}\frac{\hat{\mu}_+^{(p+1)}}{(p+1)!} \mathcal{B}_+ ,\,\,\text{and}\,\,
    \tilde{\mu}_+^\prime \coloneqq \hat{\mu}_+^\prime - h^{p}\frac{\hat{\mu}_+^{(p+1)}}{(p+1)!} \mathcal{B}_+^\prime ,
\end{align*}
where $\hat{\mu}_+^{(p+1)}$ are constructed using the $q$-th order local polynomial regression:
\begin{align*}
    \hat{\mu}_+^{(p+1)} \coloneqq (p+1)!\frac{1}{b^{p+1}} e_{p+1}^\top \check{\Gamma}_+^{-1} \check{R}^\top \check{W}_+ Y/n,
\end{align*}
where $e_{p+1}\in\mathbb{R}^{q+1}$ is the $(p+2)$-th standard basis vector, and the checked quantities (e.g., $\check{\Gamma}_+$) are defined by replacing $p$ with $q$ and $h$ with $b$.
Then, Lemma S.A.4 of \cite{Calonico_etal:2014} tells us that $\V{\tilde{\mu}_+ | x_i,\ldots,X_n} = V_+(\Sigma_Y)$ and $\V{h\tilde{\mu}_+^\prime | x_i,\ldots,X_n} = V_+^\prime(\Sigma_Y)$ (with a slight abuse of notation), where
\begin{align*}
    V_+(\Sigma) &\coloneqq \frac{1}{n}e_0^\top \Gamma_+^{-1} (R^\top W_+ \Sigma {W}_+ {R} /n) \Gamma_+^{-1}e_0\\
    &\quad+h^{2(p+1)} \frac{1}{nb^{{2(p+1)}}}(p+1)!^2 e_{p+1}^\top \check{\Gamma}_+^{-1} (\check{R}^\top \check{W}_+ \Sigma \check{W}_+ \check{R} /n) \check{\Gamma}_+^{-1}e_{p+1} \frac{\mathcal{B}_+^2}{(p+1)!^2}\\
    &\quad-2 h^{p+1}\frac{1}{nb^{p+1}}(p+1)! e_0^\top \Gamma_+^{-1} (R^\top W_+ \Sigma \check{W}_+ \check{R} /n)\check{\Gamma}_+^{-1} e_{p+1}\frac{\mathcal{B}_+}{(p+1)!}\\
    &=
    \frac{1}{n}e_0^\top \Gamma_+^{-1} (R^\top W_+ \Sigma {W}_+ {R} /n) \Gamma_+^{-1}e_0\\
    &\quad+ \frac{h^{2(p+1)}}{nb^{2(p+1)}} e_{p+1}^\top \check{\Gamma}_+^{-1} (\check{R}^\top \check{W}_+ \Sigma \check{W}_+ \check{R} /n) \check{\Gamma}_+^{-1}e_{p+1} \mathcal{B}_+^2\\
    &\quad-2 \frac{h^{p+1}}{nb^{p+1}} e_0^\top \Gamma_+^{-1} (R^\top W_+ \Sigma \check{W}_+ \check{R} /n)\check{\Gamma}_+^{-1} e_{p+1} \mathcal{B}_+,\\
    V_+^\prime(\Sigma) &\coloneqq \frac{1}{n}e_1^\top \Gamma_+^{-1} (R^\top W_+ \Sigma {W}_+ {R} /n) \Gamma_+^{-1}e_1\\
    &\quad+ \frac{h^{2(p+1)}}{nb^{2(p+1)}} e_{p+1}^\top \check{\Gamma}_+^{-1} (\check{R}^\top \check{W}_+ \Sigma \check{W}_+ \check{R} /n) \check{\Gamma}_+^{-1}e_{p+1} (\mathcal{B}_+^\prime)^2\\
    &\quad-2 \frac{h^{p+1}}{nb^{p+1}} e_1^\top \Gamma_+^{-1} (R^\top W_+ \Sigma \check{W}_+ \check{R} /n)\check{\Gamma}_+^{-1} e_{p+1} \mathcal{B}_+^\prime,
\end{align*}
and $\Sigma_Y \coloneqq \mathrm{diag}(\V{Y_i|X_1},\ldots,\V{Y_i|X_n})$.
Note that the only unknown quantity is $\Sigma_Y$.

We can derive a similar representation for the covariance. Note that
\begin{align*}
    &\C{\tilde{\mu}_+, h\tilde{\mu}_+^\prime | X_1,\ldots, X_n}=\C{\hat{\mu}_+ - h^{p+1} \hat{B}_+, h\left(\hat{\mu}_+^\prime - h^{p} \hat{B}_+^\prime\right) | X_1,\ldots, X_n}\\
    &= h\C{\hat{\mu}_+, \hat{\mu}_+^\prime | X_1,\ldots, X_n} 
    -h^{p+1}\C{\hat{\mu}_+, \hat{B}_+^\prime | X_1,\ldots, X_n} \\
    &\quad-h^{p+2}\C{\hat{B}_+, \hat{\mu}_+^\prime | X_1,\ldots, X_n}
    +h^{2(p+1)}\C{\hat{B}_+, \hat{B}_+^\prime | X_1,\ldots, X_n},
\end{align*}
where $\hat{B}_+ \coloneqq \hat{\mu}_+^{(p+1)}/(p+1)! \times \mathcal{B}_+$ and $\hat{B}_+^\prime \coloneqq \hat{\mu}_+^{(p+1)}/(p+1)! \times \mathcal{B}_+^\prime$.
Then we can show that $\C{\tilde{\mu}_+, h\tilde{\mu}_+^\prime | X_1,\ldots, X_n} = C_+(\Sigma_Y)$, where
\begin{align*}
    C_+(\Sigma) &\coloneqq \frac{1}{n}e_0^\top \Gamma_+^{-1} (R^\top W_+ \Sigma {W}_+ {R} /n) \Gamma_+^{-1}e_1\\
    &\quad - \frac{h^{p+1}}{n b^{p+1}} e_0^\top \Gamma_+^{-1} (R^\top W_+ \Sigma \check{W}_+ \check{R} /n) \check{\Gamma}_+^{-1} e_{p+1} \mathcal{B}_+^\prime \\
    &\quad - \frac{h^{p+1}}{n b^{p+1}} e_1^\top \Gamma_+^{-1} (R^\top W_+ \Sigma \check{W}_+ \check{R} /n) \check{\Gamma}_+^{-1} e_{p+1} \mathcal{B}_+\\
    &\quad+ \frac{h^{2(p+1)}}{n b^{2(p+1)}} e_{p+1}^\top \check{\Gamma}_+^{-1} (\check{R}^\top \check{W}_+ \Sigma \check{W}_+ \check{R} /n) \check{\Gamma}_+^{-1}e_{p+1} \mathcal{B}_+\mathcal{B}_+^\prime.
\end{align*}
Then, by the Cramer-Wold device, it suffices to show that 
\begin{align*}
    \frac{a_0 \left(\tilde{\mu}_+ - {\mu}_+\right) + a_1 h \left(\tilde{\mu}_+^\prime - {\mu}_+^\prime\right)}{\sqrt{a_0^2 V_+ + 2a_0 a_1 C_+ + a_1^2 V_+^\prime}}
    \to_d \mathcal{N}\left(0, 1 \right)
\end{align*}
for any $(a_0,a_1)^\top\in\mathbb{R}^2\backslash\{(0,0)^\top\}$. 
Note that the left-hand side is well defined because $\Omega$ is nonsingular by assumption.
First, the same manipulation as Lemma S.A.4 of \citet[p.~16 of Supplement]{Calonico_etal:2014} shows that
\begin{align*}
    &\frac{a_0 \left(\tilde{\mu}_+ - {\mu}_+\right) + a_1 h \left(\tilde{\mu}_+^\prime - {\mu}_+^\prime\right) }{\sqrt{a_0^2 V_+ + 2a_0 a_1 C_+ + a_1^2 V_+^\prime}} \\
    &\quad= 
    \underbrace{\frac{a_0 \left(\tilde{\mu}_+ - \E{\tilde{\mu}_+ | X_1,\ldots,X_n}\right) + a_1 h \left(\tilde{\mu}_+^\prime - \E{\tilde{\mu}_+^\prime | X_1,\ldots,X_n}\right) }{\sqrt{a_0^2 V_+ + 2a_0 a_1 C_+ + a_1^2 V_+^\prime}}}_{\eqqcolon \xi} + o_p(1).
\end{align*}
Second, we will show below that $\xi = \xi_* + o_p(1)$, where $\xi_* \coloneqq \sum_{i=1}^{n}\omega_{n,i} \varepsilon_i$, $\omega_{n,i} \coloneqq (a_0\omega_{0,n,i} + a_1 h \omega_{1,n,i})/\sqrt{\omega_{2,n}}$, $\varepsilon_i = Y_i - \mu(X_i)$, $\mu(X)\coloneqq \E{Y|X}$,
\begin{align*}
    \omega_{0,n,i} &\coloneqq e_0^\top \Gamma_{+,\star}^{-1} \frac{\mathbf{1}_i}{nh}K\left(\frac{X_i}{h}\right) r_p\left(\frac{X_i}{h}\right)
    - \frac{h^{p+1}}{b^{p+1}}e_0^\top \Gamma_{+,\star}^{-1} \vartheta_{+,\star} e_{p+1}^\top \check{\Gamma}_{+,\star}^{-1} \frac{\mathbf{1}_i}{nb} K\left(\frac{X_i}{b}\right) r_q\left(\frac{X_i}{b}\right),\\
    \omega_{1,n,i} &\coloneqq 
    e_1^\top \Gamma_{+,\star}^{-1} \frac{\mathbf{1}_i}{nh^2} K\left(\frac{X_i}{h}\right) r_p\left(\frac{X_i}{h}\right)
    - \frac{h^{p}}{b^{p+1}}e_1^\top \Gamma_{+,\star}^{-1} \vartheta_{+,\star} e_{p+1}^\top \check{\Gamma}_{+,\star}^{-1} \frac{\mathbf{1}_i}{nb} K\left(\frac{X_i}{b}\right) r_q\left(\frac{X_i}{b}\right),
\end{align*}
$\mathbf{1}_i \coloneqq \mathbf{1}\{X_i\geq0\}$, $\Gamma_{+,\star} \coloneqq \int_0^1 K(u) r_p(u) r_p(u)^\top f(uh)\,du$, $\check{\Gamma}_{+,\star} \coloneqq \int_0^1 K(u) r_q(u) r_q(u)^\top f(ub)\,du$, and $\vartheta_{+,\star} \coloneqq \int_0^1 K(u) r_p(u) u^{p+1} f(uh)\,du$; and 
\begin{align*}
    \omega_{2,n} &= \sum_{i=1}^{n} \E{(a_0\omega_{0,n,i} + a_1 h \omega_{1,n,i})^2 \varepsilon_i^2}.
\end{align*}
To see this, first observe that Lemma S.A.4 of \citet[p.~17 of Supplement]{Calonico_etal:2014} implies that
\begin{align*}
    &a_0 \left(\tilde{\mu}_+ - \E{\tilde{\mu}_+ | X_1,\ldots,X_n}\right) + a_1 h \left(\tilde{\mu}_+^\prime - \E{\tilde{\mu}_+^\prime | X_1,\ldots,X_n} \right)\\
    &\quad= \sum_{i=1}^{n}(a_0\omega_{0,n,i} + a_1 h \omega_{1,n,i}) \varepsilon_i + o_p\left((a_0^2 V_+ + 2a_0 a_1 C_+ + a_1^2 V_+^\prime)^{1/2}\right).
\end{align*}
Further, we have that
\begin{align}
    a_0^2 V_+ + 2a_0 a_1 C_+ + a_1^2 V_+^\prime &= 
    \underbrace{\sum_{i=1}^{n} (a_0\omega_{0,n,i} + a_1 h \omega_{1,n,i})^2 \E{\varepsilon_i^2 \mid X_i}}_{\eqqcolon S_n}(1+o_p(1))\notag\\
    &=\sum_{i=1}^{n}\E{(a_0\omega_{0,n,i} + a_1 h \omega_{1,n,i})^2\varepsilon_i^2}(1+o_p(1)),\label{eq: nonsingular}
\end{align}
where the second equality uses $S_n = \E{S_n}(1+o_p(1))$, which follows from $\V{S_n}/\E{S_n}^2$ $=o_p(1)$ and the Chebyshev inequality.
Hence, we obtain that $\xi = \xi_* + o_p(1)$.
Again, using Lemma S.A.4 of \cite{Calonico_etal:2014}, we obtain that $\sum_{i=1}^n \E{|\omega_{0,n,i} \varepsilon_i/\sqrt{\omega_{2,n}}|^4}=o(1)$ and that $\sum_{i=1}^n \E{|h\omega_{1,n,i} \varepsilon_i/\sqrt{\omega_{2,n}}|^4}=o(1)$.
Hence, using the elementary inequality $|x+y|^4\leq 8(|x|^4+|y|^4)$, we obtain that $\sum_{i=1}^n \E{|\omega_{n,i} \varepsilon_i|^4} = o(1)$. The Liapunov central limit theorem shows $\xi_*\to_d \mathcal{N}(0,1)$.
By the random sampling assumption, we obtain the desired result with $\mathrm{V}_{\mathtt{SRD}} (\Sigma_Y)\coloneqq V_+(\Sigma_Y) + V_-(\Sigma_Y)$, $\mathrm{V}_{\mathtt{SRD}}^\prime(\Sigma_Y) \coloneqq V_+^\prime(\Sigma_Y) + V_-^\prime(\Sigma_Y)$, and $\mathrm{C}_{\mathtt{SRD}} \coloneqq \mathrm{C}_{\mathtt{SRD}}(\Sigma_Y) = C_+(\Sigma_Y) + C_-(\Sigma_Y)$. 

Finally, we verify that $\Omega$ is asymptotically invertible.
Recalling \eqref{eq: nonsingular}, we have that $a_0^2 V_+ + 2a_0 a_1 C_+ + a_1^2 V_+^\prime = \sum_{i=1}^{n}\E{(a_0\omega_{0,n,i} + a_1 h \omega_{1,n,i})^2 \sigma_+^2(X_i)}(1 + o_p(1))$. By Assumption 1-(iii), it suffices to see that $a_0\omega_{0,n,i} + a_1 h \omega_{1,n,i}$ is not identically zero. To see this, note that
\begin{align*}
    a_0\omega_{0,n,i} + a_1 h \omega_{1,n,i} 
    = \alpha_0^\top \frac{\mathbf{1}_i}{nh}K\left(\frac{X_i}{h}\right) r_p\left(\frac{X_i}{h}\right) -
    \alpha_1^\top \frac{h^{p+1}}{b^{p+1}} \frac{\mathbf{1}_i}{nb} K\left(\frac{X_i}{b}\right) r_q\left(\frac{X_i}{b}\right)
\end{align*}
where $\alpha_0^\top \coloneqq (a_0 e_0^\top + a_1 e_1^\top)\Gamma_{+,\star}^{-1}$ and $\alpha_1^\top \coloneqq \alpha_0^\top \vartheta_{+,\star} e_{p+1}^\top \check{\Gamma}_{+,\star}^{-1}$.
We first consider the case where $h=b$.
Suppose, toward a contradiction, that the left-hand side is identically zero. Then, the $(p+2)$-th element of $\alpha_1$, i.e., $\alpha_1^\top e_{p+1}$, has to be zero. However, as the matrix $\check{\Gamma}_{+,\star}$ is positive definite, $e_{p+1}^\top \check{\Gamma}_{+,\star}^{-1}e_{p+1}>0$ holds, which in turn implies that $\alpha_0^\top \vartheta_{+,\star} =0$. Hence, $\alpha_1 = (0,\ldots,0)^\top$, thereby $\alpha_0 = (0,\ldots,0)^\top$ should hold. Since $\Gamma_{+,\star}^{-1}$ is invertible, this holds only when $(a_0, a_1)= (0,0)$, which contradicts our assumption. When $h/b\to\rho$ with $\rho<1$, noting that $K(x/h)=0$ on $x\in(h,b]$ for large $n$ enough, we can show that $\alpha_1=(0,\ldots,0)^\top$, thereby $\alpha_0=(0,\ldots,0)^\top$. The case where $\rho>1$ is similar. Hence, we obtaint the desired result.

\begin{proof}[Proof of Lemma 1]
    This lemma is a special case of the argument above.
\end{proof}
\begin{proof}[Proof of Proposition 1]
    This proposition is a corollary to the previous lemma.
\end{proof}

\subsection{Proofs of Lemmas 2, 3, and Proposition 2}
\begin{proof}[Proof of Lemma 2]
    Under the differentiability condition, $\tau_{\mathtt{SRD}}$ and $\tau_{\mathtt{SRD}}^\prime$ are both identified (\citealp[Theorem 1]{Dong_Lewbel:2015}).
    Then, a linear function over an interval including $x=0$ is identified.
    Hence, Assumption 3 implies the identification.
\end{proof}

\begin{proof}[Proof of Proposition 2]
    Lemma 1 have shown that $\tilde{\Delta}(\tau_{\mathtt{SRD}}, \tau_{\mathtt{SRD}}^\prime)\to_d \mathcal{N}((0,0)^\top, \Omega_h)$.
    Then, we have that
    \begin{align*}
        \frac{(1,x)\tilde{\Delta}(\tau_{\mathtt{SRD}}, \tau_{\mathtt{SRD}}^\prime)}{\sqrt{(1,x){\Omega}_h (1,x)^\top}}\to_d
        v(x)^\top Z,\,\text{ where }\, Z\sim\mathcal{N}\left((0,0)^\top, \mathrm{diag}(1,1)
        \right)
    \end{align*}
    and $v(x) \coloneqq {\Omega}_h^{1/2}(1,x)^\top/||{\Omega}_h^{1/2}(1,x)^\top||$.
    We are interested in the limit behavior of the supremum of the left-hand side, $\sup_{x\in[-\delta_1,\delta_2]}|{(1,x)\tilde{\Delta}(\tau_{\mathtt{SRD}}, \tau_{\mathtt{SRD}}^\prime)}/{\sqrt{(1,x){\Omega}_h (1,x)^\top}}|$. By the continuous mapping, its limit is $S\coloneqq\sup_{x\in[-\delta_1, \delta_2]} |v(x)^\top Z|$.
    By transforming to polar coordinates, we can decompose the normal variable $Z$ as $Z = R(\cos \phi, \sin \phi)^\top$, where $R$ follows the standard Rayleigh distribution, $\phi\sim \mathrm{Uniform}[0,2\pi)$, and $R\indep \phi$. Hence, we have $\sup_{x\in[-\delta_1, \delta_2]} |v(x)^\top Z| = R \cdot \sup_{x\in[-\delta_1, \delta_2]}|v(x)^\top (\cos \phi, \sin \phi)^\top|$. 
    
    Now, recalling that the definition of $v(x)$, it is on the unit circle. 
    To study the direction of $v(x)$, we define the non-normalized version $u(x) \coloneqq {\Omega}_h^{1/2}(1,x)^\top$, and identify $\mathbb{R}^2$ with the complex plane so that we examine the argument of the complex number $u_1(x) + i u_2(x)$, where $u(x) = (u_1(x), u_2(x))^\top$.
    By construction, there exist real constants $a_1, a_2, b_1, b_2$ such that $u(x)=(a_1 + b_1 x, a_2 + b_2 x)^\top$.
    Then we compute:
    \begin{align*}
        \frac{\partial}{\partial x} \arg\left( u_1(x) + iu_2(x) \right)
        = \frac{a_1 b_2 - a_2 b_1}
                {(a_1 + b_1 x)^2 + (a_2 + b_2 x)^2}>0,
    \end{align*}
    since $a_1 b_2 - a_2 b_1 = \det({\Omega}_h^{1/2}) > 0$ by the regularity of ${\Omega}_h$. Hence, the argument of $u(x)$ is strictly increasing in $x$, and so is $v(x)$.
    Therefore, as $x$ ranges over $[-\delta_1, \delta_2]$, $v(x)$ traces out an arc of the unit circle.
    Let $\theta(x)$ denote this argument, and it is easy to see that there exists some closed interval $[\theta_l, \theta_u]$ such that $\theta(x) \in [\theta_l, \theta_u]$, $\theta_u - \theta_l \leq \pi$, and the boundary points are attained at $x=-\delta_1$ or $x=\delta_2$.
    Then, we can rewrite that
    \begin{align*}
        S&=R \cdot \sup_{x\in[-\delta_1, \delta_2]}|v(x)^\top (\cos \phi, \sin \phi)^\top| \\
        &= R \cdot \sup_{x\in[-\delta_1, \delta_2]}|\cos(\theta(x) - \phi)|\\
        &= R\cdot \sup_{\theta\in [\theta_l, \theta_u]}|\cos(\theta - \phi)|,
    \end{align*}
    where the second equality uses the cosine addition formula.
    The distribution of the last term is equal to that of $R\cdot \sup_{\theta\in [0, \theta_u - \theta_l]}|\cos(\theta - \phi)|$, whose distribution is again equivalent to the distribution of $R\cdot \sup_{\theta\in [0, \theta_u - \theta_l]}|\cos(\theta - \varphi)|$, where $\varphi\sim \mathrm{Uniform}[0,\pi]$.
    Then, we can further rewrite this as
    \begin{align*}
        R\cdot\begin{cases}
            1 & \text{ if } \varphi \in [0,\theta_u - \theta_l]\eqqcolon \mathrm{I}_1\\
            |\cos((\theta_u-\theta_l)-\varphi)| & \text{ if } \varphi \in [\theta_u - \theta_l, \pi/2 + (\theta_u - \theta_l)/2]\eqqcolon \mathrm{I}_2\\
            |\cos(-\varphi)| & \text{ if } \varphi \in [\pi/2 + (\theta_u - \theta_l)/2, \pi]\eqqcolon \mathrm{I}_3.
        \end{cases}
    \end{align*}
    Figure \ref{fig: proof} must be instructive to derive this. Again, by observing Figure \ref{fig: proof}, we finally obtain that
    \begin{align*}
        S &= R\cdot\begin{cases}
            1 & \text{ with probability } (\theta_u - \theta_l)/\pi\\
            \cos(U) & \text{ with probability } 1-(\theta_u - \theta_l)/\pi,
        \end{cases}\\
        U &\sim \mathrm{Uniform}[0,\pi/2 - (\theta_u - \theta_l)/2],\,\,R\indep U.
    \end{align*}
    \begin{figure}[t]
        \centering
        \includegraphics[width=0.8\linewidth]{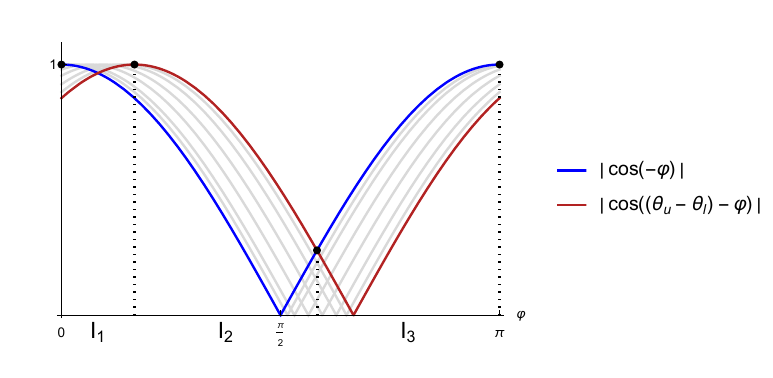}
        \caption{Illustration of $\mathrm{I}_1$, $\mathrm{I}_2$, and $\mathrm{I}_3$}
        \label{fig: proof}
    \end{figure}
    Hence, the distribution function of $S$ is given by
    \begin{align*}
        &\P{S\leq s}\\ &=
        \frac{\theta_u - \theta_l}{\pi}\P{R\leq s} 
        + \left(1-\frac{\theta_u - \theta_l}{\pi}\right) \int_0^{\pi/2 - (\theta_u - \theta_l)/2}\P{R\leq \frac{s}{\cos(u)}}\frac{2}{\pi-(\theta_u-\theta_l)}\,du\\
        &=\frac{\theta_u - \theta_l}{\pi} \left[1-\exp\left(-\frac{s^2}{2}\right)\right] + 
        \frac{2}{\pi}\int_0^{\pi/2 - (\theta_u - \theta_l)/2}1-\exp\left(-\frac{s^2}{2\cos(u)^2}\right)\,du.
    \end{align*}
    Thus, combined with the consistency of $\hat{\Omega}$, the critical value defined in the proposition provides asymptotically correct coverage.
\end{proof}

\begin{proof}[Proof of Lemma 3]
    If $[-\delta_1, \delta_2]\supseteq [-\delta_1^\prime, \delta_2^\prime]$, corresponding $\ell$ satisfies $\ell \geq \ell^\prime$. Then, it suffices to see that the critical value is nondecreasing in $\ell$.
    Rewrite
    \begin{align*}
        P(s;\ell)=\frac{\ell}{\pi} \left[1-\exp\left(-\frac{s^2}{2}\right)\right] + 
        \frac{2}{\pi}\int_0^{\pi/2 - \ell/2}1-\exp\left(-\frac{s^2}{2\cos(u)^2}\right)\,du.
    \end{align*}
    By the Leibniz rule, we can show that $\partial P(s;\ell) / \partial \ell\leq 0$. It is also straightforward to see that $\partial P(s;\ell) / \partial s > 0$ for $s>0$.
    Then, the implicit function theorem proves the statement.
\end{proof}

\subsection{Fuzzy RD Design}
Extension of our proposals to the fuzzy RD case is straightforward.
Suppose a random sample $(Y_i(0), Y_i(1), T_i(0), T_i(1), X_i)$.
Here, $T(0)$ and $T(1)$ denote the potential treatment indicator, and we observe $T_i = T_i(1)\mathbf{1}\{X_i \geq 0\} + T_i(0)\mathbf{1}\{X_i < 0\}$.
In fuzzy designs, our target parameters are $\tau_{\mathtt{FRD}}\coloneqq \E{Y_i(1) - Y_i(0) | X_i = 0, T_i(1)> T_i(0)}$ and its derivative with respect to $x$, $\tau_{\mathtt{FRD}}^\prime$.
Under Assumption A3 of \cite{Dong_Lewbel:2015}, these are identified as
\begin{align*}
\tau_{\mathtt{FRD}}
= \frac{\tau_Y}{\tau_T}
\coloneqq \frac{\mu_{Y,+}-\mu_{Y,-}}{\mu_{T,+}-\mu_{T,-}},\quad
\tau_{\mathtt{FRD}}^\prime
= 
\frac{\tau_Y^\prime}{\tau_T}-\frac{\tau_T^\prime}{\tau_T}\cdot\frac{\tau_Y}{\tau_T}
= \frac{\tau_Y^\prime\tau_T-\tau_Y\tau_T^\prime}{\tau_T^{2}},
\end{align*}
where $\mu_{A,+}\coloneqq \lim_{x\downarrow0}\mu_{A}(x)$, $\mu_{A,-}\coloneqq \lim_{x\uparrow0}\mu_{A}(x)$, and $\mu_A(x)\coloneqq \E{A_i | X_i=x}$, where $A$ equals either $Y$ or $T$.

In addition to the assumptions made in Lemma 1, we assume the following:
\begin{assumption}\label{assumption: fuzzy}
    $h\to0$. Furthermore, in a neighbourhood $(-\kappa_0,\kappa_0)$ around the cutoff:
    \item[(i)] $\mu_{T,+}(x)\coloneqq\E{T_i(1)|X_i=x}$ and $\mu_{T,-}(x)\coloneqq\E{T_i(0)|X_i=x}$ are $S$-times continuously differentiable.
    \item[(ii)] $\sigma_{T,+}^2(x)\coloneqq \V{T_i(1)|X_i=x}$ and $\sigma_{T,-}^2(x)\coloneqq \V{T_i(0)|X_i=x}$ are continuous and bounded away from zero.
\end{assumption}

Below, we suppose that the assumptions made in Section \ref{sec:proof1} and Assumption \ref{assumption: fuzzy} are satisfied.
Let $\hat{\tau}_{\mathtt{FRD}}$ and $\hat{\tau}_{\mathtt{FRD}}^\prime$ be based on the $p$-th order local polynomial estimators using the same bandwidth $h$ for every estimation step.
The bias-corrected estimator of ${\tau}_{\mathtt{FRD}}$ is given by (\citealp[Supplement, p.24]{Calonico_etal:2014})
\begin{align*}
    \tilde{\tau}_{\mathtt{FRD}} &= \hat{\tau}_{\mathtt{FRD}} - 
    h^{p+1}\left[\left\{\frac{1}{\hat{\tau}_{T}}\frac{\hat{\mu}_{Y,+}^{(p+1)}}{(p+1)!} - \frac{\hat{\tau}_{Y}}{\hat{\tau}_{T}^2}\frac{\hat{\mu}_{T,+}^{(p+1)}}{(p+1)!}\right\}\mathcal{B}_+
    -
    \left\{\frac{1}{\hat{\tau}_{T}}\frac{\hat{\mu}_{Y,-}^{(p+1)}}{(p+1)!} - \frac{\hat{\tau}_{Y}}{\hat{\tau}_{T}^2}\frac{\hat{\mu}_{T,-}^{(p+1)}}{(p+1)!}\right\}\mathcal{B}_-\right].
\end{align*}
Similarly, we can define the bias-corrected estimator of ${\tau}_{\mathtt{FRD}}^\prime$ by
\begin{align*}
    \tilde{\tau}_{\mathtt{FRD}}^\prime &= \hat{\tau}_{\mathtt{FRD}}^\prime - 
    h^{p}\left[\left\{\frac{1}{\hat{\tau}_{T}}\frac{\hat{\mu}_{Y,+}^{(p+1)}}{(p+1)!} - \frac{\hat{\tau}_{Y}}{\hat{\tau}_{T}^2}\frac{\hat{\mu}_{T,+}^{(p+1)}}{(p+1)!}\right\}\mathcal{B}_+^\prime
    -
    \left\{\frac{1}{\hat{\tau}_{T}}\frac{\hat{\mu}_{Y,-}^{(p+1)}}{(p+1)!} - \frac{\hat{\tau}_{Y}}{\hat{\tau}_{T}^2}\frac{\hat{\mu}_{T,-}^{(p+1)}}{(p+1)!}\right\}\mathcal{B}_-^\prime\right].
\end{align*}
It is useful to rewrite (\citealp[Supplement, p.~24]{Calonico_etal:2014}):
\begin{align*}
    \tilde{\tau}_{\mathtt{FRD}} - \tau_{\mathtt{FRD}} = 
    \underbrace{\frac{1}{\tau_{T}} \left(\tilde{\tau}_{Y} - \tau_Y\right)
    - \frac{\tau_{Y}}{\tau_{T}^2} \left(\tilde{\tau}_{T} - \tau_T\right)}_{\eqqcolon \mathring{\tau}_{\mathtt{FRD}}}
    + R_{n} - R_{n}^{\mathtt{bc}},
\end{align*}
where $R_{n}$ and $R_{n}^{\mathtt{bc}}$ are given in \citet[Supplement, p.~24]{Calonico_etal:2014} and are shown to be asymptotically negligible. 
Similarly, one can show that
\begin{align}
    \tilde{\tau}_{\mathtt{FRD}}^\prime - {\tau}_{\mathtt{FRD}}^\prime &= \underbrace{\frac{1}{\tau_T}\left(\tilde{\tau}_Y^\prime-\tau_Y^\prime\right)
    -\frac{\tau_Y}{\tau_T^2}\left(\tilde{\tau}_T^\prime-\tau_T^\prime\right)}_{\eqqcolon \mathring{\tau}_{\mathtt{FRD}}^\prime}
    + R_{n}^\prime - (R_{n}^\prime)^{\mathtt{bc}},\notag
\end{align}
where $(R_{n}^\prime)^{\mathtt{bc}}$ is analogously defined as $R_{n}^{\mathtt{bc}}$, and $R_{n}^\prime$ is given by
\begin{align*}
    R_{n}^\prime \coloneqq \frac{\hat{\tau}_Y^\prime}{\hat{\tau}_T{\tau}_T}\left({\tau}_T-\hat{\tau}_T\right) -
    \frac{\hat{\tau}_T^\prime\hat{\tau}_Y}{\hat{\tau}_T^{2}\tau_T^{2}}\,(\tau_T^{2}-\hat{\tau}_T^{2}) -
    \frac{\hat{\tau}_T^\prime}{\tau_T^{2}}(\hat{\tau}_Y-\tau_Y)
\end{align*}
Then, $(R_{n}^\prime)^{\mathtt{bc}}$ and $R_{n}^\prime$ can be shown to be of smaller order analogously.
The conditional variance of $\mathring{\tau}_{\mathtt{FRD}}$ is derived by \citet[Theorem A.2]{Calonico_etal:2014}, which is written by $\mathrm{V}_{\mathtt{FRD}}$:
\begin{align*}
    \mathrm{V}_{\mathtt{FRD}} \coloneqq 
    \frac{1}{\tau_T^2}\mathrm{V}_{\mathtt{SRD}}(\Sigma_Y)
    -2\frac{\tau_Y}{\tau_T^3}\mathrm{V}_{\mathtt{SRD}}(\Sigma_{Y,T})
    +\frac{\tau_Y^2}{\tau_T^4}\mathrm{V}_{\mathtt{SRD}}(\Sigma_T)
\end{align*}
where $\Sigma_{Y,T} \coloneqq \mathrm{diag}(\C{Y_i, T_i | X_i}, \ldots, \C{Y_i, T_i | X_n})$ and $\Sigma_T$ is defined analogously to $\Sigma_Y$.
The conditional variance of $h\mathring{\tau}_{\mathtt{FRD}}^\prime$ can be similarly obtained as provided below:
\begin{align*}
    \mathrm{V}_{\mathtt{FRD}}^\prime \coloneqq 
    \frac{1}{\tau_T^2}\mathrm{V}_{\mathtt{SRD}}^\prime(\Sigma_Y)
    -2\frac{\tau_Y}{\tau_T^3}\mathrm{V}_{\mathtt{SRD}}^\prime(\Sigma_{Y,T})
    +\frac{\tau_Y^2}{\tau_T^4}\mathrm{V}_{\mathtt{SRD}}^\prime(\Sigma_T).
\end{align*}
Observing that $\mathring{\tau}_{\mathtt{FRD}}^{(\nu)}$ is a (weighted) difference of the two bias-corrected sharp RD estimands, the covariance can be computed as $\C{\mathring{\tau}_{\mathtt{FRD}}, h\mathring{\tau}_{\mathtt{FRD}}^\prime | X_1,\ldots,X_n} = \mathrm{C}_{\mathtt{FRD}}$, where
\begin{align*}
    \mathrm{C}_{\mathtt{FRD}}\coloneqq
    \frac{1}{\tau_T^2}\mathrm{C}_{\mathtt{SRD}}(\Sigma_Y)
    - \frac{\tau_Y}{\tau_T^3}\mathrm{C}_{\mathtt{SRD}}(\Sigma_{Y,T})
    - \frac{\tau_Y}{\tau_T^3}\mathrm{C}_{\mathtt{SRD}}(\Sigma_{T,Y})
    + \frac{\tau_Y^2}{\tau_T^4}\mathrm{C}_{\mathtt{SRD}}(\Sigma_{T}),
\end{align*}
Given this result, a confidence region analogous to that in Proposition 1 can be readily constructed:
\begin{align*}
    \Omega_{\mathtt{FRD}}^{-1/2}\mathrm{diag}(1,h)\tilde{\Delta}(\tau_{\mathtt{FRD}}, \tau_{\mathtt{FRD}}^\prime)\to_d
    \mathcal{N}\left((0,0)^\top, \mathrm{diag}(1,1)
    \right),\quad
    \Omega_{\mathtt{FRD}}\coloneqq\begin{pmatrix}
        \mathrm{V}_{\mathtt{FRD}} & \mathrm{C}_{\mathtt{FRD}}\\
        \mathrm{C}_{\mathtt{FRD}} & \mathrm{V}_{\mathtt{FRD}}^\prime
    \end{pmatrix}.
\end{align*}
Instead of Assumption 3 in the main text, we can assume the local linearity of $\mathbb{E}[Y_i(1) - Y_i(0) | X_i = x, T_i(1)> T_i(0)]$ in $x\in[-\delta_1,\delta_2]$, thereby obtaining an identification result and a confidence band analogous to those in Lemma 2 and Proposition 2.

\section{Simulation}
The data-generating processes used in the simulation studies are provided in Figure \ref{fig: simulation}. Additional simulation results are summarized in Tables \ref{tab: simulation2}-\ref{tab: simulation2-m}.
\begin{figure}[H]
    \begin{center}
    \begin{subfigure}[b]{0.32\textwidth}
        \centering
        \includegraphics[width=\linewidth]{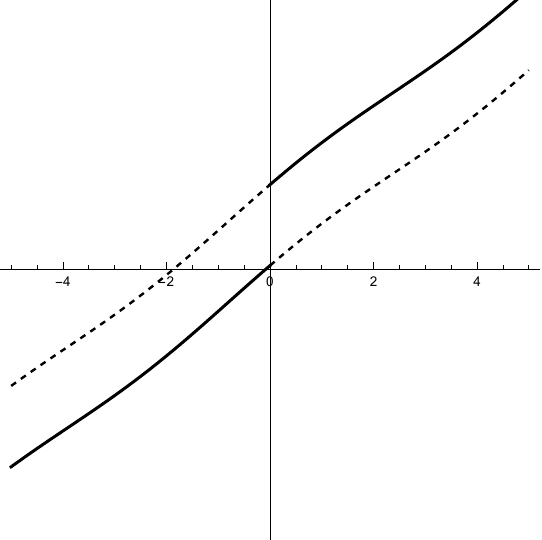}
        \caption{DGP 1}
        \label{fig: simu1a}
    \end{subfigure}
    \hfill
    \begin{subfigure}[b]{0.32\textwidth}
        \centering
        \includegraphics[width=\linewidth]{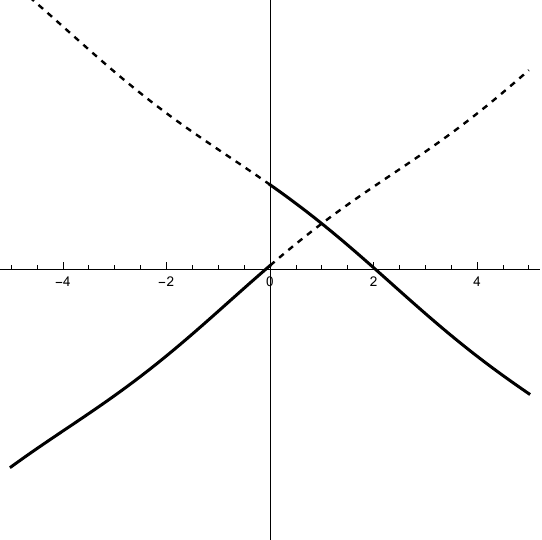}
        \caption{DGP 2}
        \label{fig: simu2a}
    \end{subfigure}
    \hfill
    \begin{subfigure}[b]{0.32\textwidth}
        \centering
        \includegraphics[width=\linewidth]{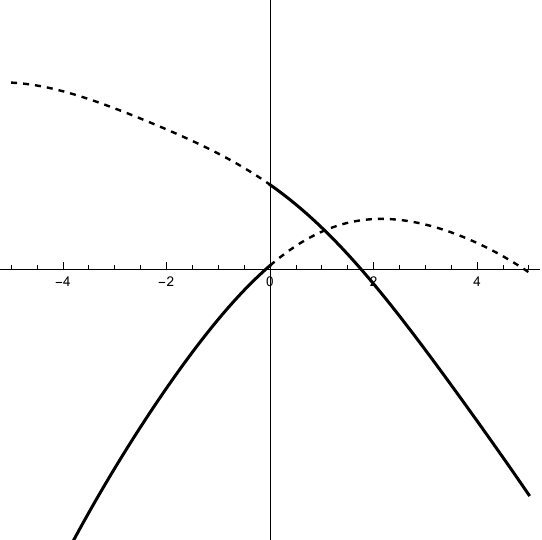}
        \caption{DGP 3}
        \label{fig: simu3a}
    \end{subfigure}

    \caption{Data Generating Processes}
    \label{fig: simulationa}
    \end{center}
\end{figure}

\begin{landscape}
\begin{table}
    \centering
    \begin{tabular}{lccccccccccc}\hline\hline
         & \multicolumn{5}{c}{Local Linear ($p=1$)} & & \multicolumn{5}{c}{Local Quadratic ($p=2$)}\\\cline{2-6}\cline{8-12}
         & $\hat{h}_{\mathtt{mean}}$ & $\hat{h}_{\mathtt{deriv}}$ & $\hat{h}_{\mathtt{both}}$ & $\hat{h}_{\mathtt{reg0}}$ & $\hat{h}_{\mathtt{reg1}}$ & & $\hat{h}_{\mathtt{mean}}$ & $\hat{h}_{\mathtt{deriv}}$ & $\hat{h}_{\mathtt{both}}$ & $\hat{h}_{\mathtt{reg0}}$ & $\hat{h}_{\mathtt{reg1}}$\\\hline
        \multicolumn{4}{l}{DGP 1: Coverage}\\
        \hspace{0.5cm}$\mathcal{R}_{0.95}$ & 0.930 & 0.928 & 0.930 & 0.930 & 0.929 & & 0.929 & 0.930 & 0.929 & 0.928 & 0.928 \\
        \hspace{0.5cm}$\mathcal{U}_{0.95}$ ($\delta=0.05$) & 0.939 & 0.941 & 0.939 & 0.938 & 0.935 & & 0.920 & 0.922 & 0.920 & 0.920 & 0.924 \\
        \hspace{0.5cm}$\mathcal{U}_{0.95}$ ($\delta=1$) & 0.931 & 0.927 & 0.927 & 0.927 & 0.929 & & 0.927 & 0.929 & 0.925 & 0.927 & 0.926 \\
        \hspace{0.5cm}$\mathcal{U}_{0.95}$ ($\delta=5$) & 0.930 & 0.928 & 0.928 & 0.930 & 0.929 & & 0.929 & 0.930 & 0.930 & 0.928 & 0.928 \\
        \multicolumn{4}{l}{DGP 1: MISE}\\
        \hspace{0.5cm}$\delta=0.05$ & 0.027 & 0.029 & 0.027 & 0.019 & 0.020 & & 0.099 & 0.108 & 0.098 & 0.032 & 0.032 \\
        \hspace{0.5cm}$\delta=1$ & 2.045 & 1.840 & 1.556 & 1.292 & 1.513 & & 7.115 & 6.112 & 4.903 & 3.409 & 3.746 \\
        \hspace{0.5cm}$\delta=5$ & 190.9 & 159.8 & 145.1 & 116.1 & 141.7 & & 653.4 & 506.2 & 473.4 & 349.8 & 391.3 \\\hline
        
        \multicolumn{4}{l}{DGP 2: Coverage}\\
        \hspace{0.5cm}$\mathcal{R}_{0.95}$ & 0.934 & 0.930 & 0.933 & 0.927 & 0.929 & & 0.934 & 0.932 & 0.933 & 0.934 & 0.935 \\
        \hspace{0.5cm}$\mathcal{U}_{0.95}$ ($\delta=0.05$) & 0.946 & 0.944 & 0.946 & 0.933 & 0.936 & & 0.930 & 0.926 & 0.929 & 0.925 & 0.929 \\
        \hspace{0.5cm}$\mathcal{U}_{0.95}$ ($\delta=1$) & 0.932 & 0.926 & 0.929 & 0.927 & 0.927 & & 0.933 & 0.931 & 0.933 & 0.933 & 0.934 \\
        \hspace{0.5cm}$\mathcal{U}_{0.95}$ ($\delta=5$) & 0.934 & 0.930 & 0.931 & 0.927 & 0.929 & & 0.934 & 0.933 & 0.933 & 0.934 & 0.935 \\
        \multicolumn{4}{l}{DGP 2: MISE}\\
        \hspace{0.5cm}$\delta=0.05$ & 0.028 & 0.026 & 0.028 & 0.020 & 0.021 & & 0.086 & 0.105 & 0.084 & 0.033 & 0.033 \\
        \hspace{0.5cm}$\delta=1$ & 2.037 & 1.690 & 1.587 & 1.291 & 1.521 & & 6.524 & 6.128 & 4.847 & 3.470 & 3.887 \\
        \hspace{0.5cm}$\delta=5$ & 188.7 & 148.1 & 144.7 & 114.4 & 141.0 & & 611.5 & 515.6 & 494.5 & 356.4 & 408.3 \\\hline
        
        \multicolumn{4}{l}{DGP 3: Coverage}\\
        \hspace{0.5cm}$\mathcal{R}_{0.95}$ & 0.933 & 0.933 & 0.933 & 0.929 & 0.930 & & 0.912 & 0.923 & 0.910 & 0.915 & 0.915 \\
        \hspace{0.5cm}$\mathcal{U}_{0.95}$ ($\delta=0.05$) & 0.940 & 0.943 & 0.940 & 0.935 & 0.931 & & 0.916 & 0.923 & 0.916 & 0.910 & 0.904 \\
        \hspace{0.5cm}$\mathcal{U}_{0.95}$ ($\delta=1$) & 0.931 & 0.931 & 0.929 & 0.927 & 0.928 & & 0.911 & 0.923 & 0.920 & 0.914 & 0.915 \\
        \hspace{0.5cm}$\mathcal{U}_{0.95}$ ($\delta=5$) & 0.910 & 0.905 & 0.904 & 0.895 & 0.912 & & 0.906 & 0.910 & 0.907 & 0.901 & 0.906 \\
        \multicolumn{4}{l}{DGP 3: MISE}\\
        \hspace{0.5cm}$\delta=0.05$ & 0.026 & 0.023 & 0.026 & 0.021 & 0.022 & & 0.090 & 0.135 & 0.089 & 0.035 & 0.035 \\
        \hspace{0.5cm}$\delta=1$ & 2.042 & 1.526 & 1.403 & 1.403 & 1.619 & & 7.639 & 7.086 & 5.422 & 3.882 & 4.340 \\
        \hspace{0.5cm}$\delta=5$ & 205.7 & 148.5 & 142.6 & 138.4 & 161.8 & & 752.0 & 574.6 & 553.3 & 415.3 & 471.8 \\\hline
    \end{tabular}
    \caption{Main Simulation Results ($n=500$)}
    \label{tab: simulation2}
\end{table}
\end{landscape}

\begin{landscape}
\begin{table}
    \centering
    \begin{tabular}{lccccccccccc}\hline\hline
         & \multicolumn{5}{c}{Local Linear ($p=1$)} & & \multicolumn{5}{c}{Local Quadratic ($p=2$)}\\\cline{2-6}\cline{8-12}
         & $\hat{h}_{\mathtt{mean}}$ & $\hat{h}_{\mathtt{deriv}}$ & $\hat{h}_{\mathtt{both}}$ & $\hat{h}_{\mathtt{reg0}}$ & $\hat{h}_{\mathtt{reg1}}$ & & $\hat{h}_{\mathtt{mean}}$ & $\hat{h}_{\mathtt{deriv}}$ & $\hat{h}_{\mathtt{both}}$ & $\hat{h}_{\mathtt{reg0}}$ & $\hat{h}_{\mathtt{reg1}}$\\\hline
        \multicolumn{4}{l}{DGP 1: MSE at $x=-\delta_1$}\\
        \hspace{0.5cm}$\delta=0.05$ & 0.138 & 0.163 & 0.137 & 0.105 & 0.112 & & 0.481 & 0.584 & 0.474 & 0.176 & 0.181 \\
        \hspace{0.5cm}$\delta=1$ & 1.170 & 1.029 & 0.942 & 0.784 & 0.964 & & 4.732 & 3.897 & 3.278 & 2.503 & 2.898 \\
        \hspace{0.5cm}$\delta=5$ & 25.618 & 21.033 & 20.342 & 16.609 & 20.619 & & 105.528 & 78.318 & 74.746 & 57.426 & 66.819 \\
        \multicolumn{4}{l}{DGP 1: MSE at $x=\delta_2$}\\
        \hspace{0.5cm}$\delta=0.05$ & 0.135 & 0.159 & 0.135 & 0.103 & 0.107 & & 0.474 & 0.556 & 0.465 & 0.171 & 0.173 \\
        \hspace{0.5cm}$\delta=1$ & 1.128 & 0.941 & 0.887 & 0.730 & 0.873 & & 4.576 & 3.336 & 3.066 & 2.391 & 2.744 \\
        \hspace{0.5cm}$\delta=5$ & 25.405 & 20.595 & 19.996 & 16.336 & 20.166 & & 104.745 & 75.513 & 73.086 & 56.862 & 66.051 \\\hline
        
        \multicolumn{4}{l}{DGP 2: MSE at $x=-\delta_1$}\\
        \hspace{0.5cm}$\delta=0.05$ & 0.132 & 0.147 & 0.132 & 0.103 & 0.109 & & 0.381 & 0.516 & 0.377 & 0.169 & 0.171 \\
        \hspace{0.5cm}$\delta=1$ & 1.116 & 0.946 & 0.878 & 0.771 & 0.924 & & 3.906 & 3.212 & 2.765 & 2.216 & 2.469 \\
        \hspace{0.5cm}$\delta=5$ & 23.901 & 19.925 & 19.159 & 16.533 & 20.175 & & 86.109 & 68.590 & 65.702 & 51.052 & 57.636 \\
        \multicolumn{4}{l}{DGP 2: MSE at $x=\delta_2$}\\
        \hspace{0.5cm}$\delta=0.05$ & 0.128 & 0.146 & 0.127 & 0.101 & 0.107 & & 0.367 & 0.519 & 0.363 & 0.167 & 0.171 \\
        \hspace{0.5cm}$\delta=1$ & 1.023 & 0.920 & 0.842 & 0.737 & 0.886 & & 3.630 & 3.266 & 2.667 & 2.171 & 2.457 \\
        \hspace{0.5cm}$\delta=5$ & 23.440 & 19.791 & 19.005 & 16.361 & 19.988 & & 84.728 & 68.857 & 65.528 & 50.826 & 57.574 \\\hline
        
        \multicolumn{4}{l}{DGP 3: MSE at $x=-\delta_1$}\\
        \hspace{0.5cm}$\delta=0.05$ & 0.142 & 0.115 & 0.139 & 0.097 & 0.106 & & 0.453 & 0.504 & 0.447 & 0.168 & 0.172 \\
        \hspace{0.5cm}$\delta=1$ & 1.232 & 0.891 & 0.816 & 0.745 & 0.924 & & 4.179 & 3.577 & 3.031 & 2.394 & 2.707 \\
        \hspace{0.5cm}$\delta=5$ & 32.871 & 24.670 & 24.193 & 22.199 & 26.828 & & 99.954 & 83.566 & 78.624 & 62.851 & 70.517 \\
        \multicolumn{4}{l}{DGP 3: MSE at $x=\delta_2$}\\
        \hspace{0.5cm}$\delta=0.05$ & 0.139 & 0.114 & 0.136 & 0.099 & 0.108 & & 0.452 & 0.506 & 0.447 & 0.171 & 0.174 \\
        \hspace{0.5cm}$\delta=1$ & 1.166 & 0.895 & 0.825 & 0.784 & 0.956 & & 4.156 & 3.614 & 2.939 & 2.445 & 2.747 \\
        \hspace{0.5cm}$\delta=5$ & 32.486 & 26.570 & 25.735 & 23.412 & 27.140 & & 99.134 & 84.356 & 78.004 & 62.250 & 69.845 \\\hline
    \end{tabular}
    \caption{Simulation Results: MSE ($n=1000$)}
    \label{tab: simulation1-m}
\end{table}
\end{landscape}

\begin{landscape}
\begin{table}
    \centering
    \begin{tabular}{lccccccccccc}\hline\hline
         & \multicolumn{5}{c}{Local Linear ($p=1$)} & & \multicolumn{5}{c}{Local Quadratic ($p=2$)}\\\cline{2-6}\cline{8-12}
         & $\hat{h}_{\mathtt{mean}}$ & $\hat{h}_{\mathtt{deriv}}$ & $\hat{h}_{\mathtt{both}}$ & $\hat{h}_{\mathtt{reg0}}$ & $\hat{h}_{\mathtt{reg1}}$ & & $\hat{h}_{\mathtt{mean}}$ & $\hat{h}_{\mathtt{deriv}}$ & $\hat{h}_{\mathtt{both}}$ & $\hat{h}_{\mathtt{reg0}}$ & $\hat{h}_{\mathtt{reg1}}$\\\hline
        \multicolumn{4}{l}{DGP 1: MSE at $x=-\delta_1$}\\
        \hspace{0.5cm}$\delta=0.05$ & 0.274 & 0.297 & 0.273 & 0.190 & 0.199 & & 0.993 & 1.080 & 0.981 & 0.326 & 0.329 \\
        \hspace{0.5cm}$\delta=1$ & 2.500 & 2.175 & 1.837 & 1.506 & 1.826 & & 8.523 & 6.847 & 5.742 & 4.434 & 4.921 \\
        \hspace{0.5cm}$\delta=5$ & 56.601 & 47.346 & 42.868 & 34.188 & 41.877 & & 193.129 & 148.859 & 139.602 & 104.086 & 116.461 \\
        \multicolumn{4}{l}{DGP 1: MSE at $x=\delta_2$}\\
        \hspace{0.5cm}$\delta=0.05$ & 0.277 & 0.297 & 0.276 & 0.195 & 0.204 & & 1.012 & 1.098 & 1.001 & 0.331 & 0.335 \\
        \hspace{0.5cm}$\delta=1$ & 2.557 & 2.175 & 1.983 & 1.613 & 1.923 & & 8.888 & 7.193 & 6.066 & 4.521 & 5.034 \\
        \hspace{0.5cm}$\delta=5$ & 56.889 & 47.347 & 43.324 & 34.725 & 42.364 & & 194.956 & 150.590 & 140.784 & 104.518 & 117.026 \\\hline
        
        \multicolumn{4}{l}{DGP 2: MSE at $x=-\delta_1$}\\
        \hspace{0.5cm}$\delta=0.05$ & 0.278 & 0.268 & 0.277 & 0.199 & 0.210 & & 0.874 & 1.070 & 0.861 & 0.334 & 0.337 \\
        \hspace{0.5cm}$\delta=1$ & 2.450 & 2.015 & 1.954 & 1.552 & 1.889 & & 8.215 & 7.343 & 5.985 & 4.588 & 5.220 \\
        \hspace{0.5cm}$\delta=5$ & 55.774 & 43.933 & 43.020 & 33.970 & 41.981 & & 182.397 & 153.794 & 147.284 & 106.407 & 122.036 \\
        \multicolumn{4}{l}{DGP 2: MSE at $x=\delta_2$}\\
        \hspace{0.5cm}$\delta=0.05$ & 0.283 & 0.267 & 0.282 & 0.199 & 0.208 & & 0.862 & 1.046 & 0.847 & 0.331 & 0.333 \\
        \hspace{0.5cm}$\delta=1$ & 2.562 & 2.003 & 1.927 & 1.539 & 1.856 & & 7.957 & 6.868 & 5.760 & 4.531 & 5.148 \\
        \hspace{0.5cm}$\delta=5$ & 56.330 & 43.872 & 42.899 & 33.905 & 41.821 & & 181.106 & 151.419 & 145.708 & 106.122 & 121.676 \\\hline
        
        \multicolumn{4}{l}{DGP 3: MSE at $x=-\delta_1$}\\
        \hspace{0.5cm}$\delta=0.05$ & 0.267 & 0.230 & 0.264 & 0.210 & 0.225 & & 0.908 & 1.360 & 0.893 & 0.355 & 0.357 \\
        \hspace{0.5cm}$\delta=1$ & 2.648 & 1.814 & 1.702 & 1.721 & 2.018 & & 9.492 & 7.885 & 6.554 & 5.190 & 5.825 \\
        \hspace{0.5cm}$\delta=5$ & 64.520 & 44.583 & 43.242 & 42.644 & 50.079 & & 226.684 & 172.470 & 166.542 & 127.670 & 144.265 \\
        \multicolumn{4}{l}{DGP 3: MSE at $x=\delta_2$}\\
        \hspace{0.5cm}$\delta=0.05$ & 0.258 & 0.231 & 0.256 & 0.207 & 0.222 & & 0.927 & 1.366 & 0.905 & 0.351 & 0.358 \\
        \hspace{0.5cm}$\delta=1$ & 2.460 & 1.868 & 1.720 & 1.680 & 1.971 & & 9.853 & 7.995 & 6.578 & 5.098 & 5.825 \\
        \hspace{0.5cm}$\delta=5$ & 62.855 & 48.600 & 46.502 & 44.561 & 51.130 & & 225.970 & 171.903 & 165.509 & 125.132 & 142.466 \\\hline
    \end{tabular}
    \caption{Simulation Results: MSE ($n=500$)}
    \label{tab: simulation2-m}
\end{table}
\end{landscape}

\bibliographystyle{apalike} 
\bibliography{refs}

\end{document}